\RequirePackage[orthodox]{nag}
\PassOptionsToPackage{nameinlink}{cleveref}

\documentclass[a4paper,USenglish,cleveref]{lipics-v2021}

\usepackage{multirow}
\usepackage{xfrac}

\usepackage{float}
\usepackage{subcaption}
\usepackage{refcount}

\usepackage{pifont}
\usepackage{epsdice}

\usepackage{placeins}

\usepackage[ruled,vlined,linesnumbered]{algorithm2e}

\usepackage{xspace}
\newcommand{\BD}{\sigma}
\newcommand{\NP}{\ensuremath{\mathsf{NP}}\xspace}
\newcommand{\p}{\ensuremath{\mathsf{P}}\xspace}

\title{Graph Spanners for Group Steiner Distances}

\author{Davide Bilò}{Department of Information Engineering, Computer Science and Mathematics,\\ University of L'Aquila, Italy}{davide.bilo@univaq.it}{https://orcid.org/0000-0003-3169-4300}{}
\author{Luciano Gualà}{Department of Enterprise Engineering, University of Rome ``Tor Vergata'', Italy}{guala@mat.uniroma2.it}{https://orcid.org/0000-0001-6976-5579}{}
\author{Stefano Leucci}{Department of Information Engineering, Computer Science and Mathematics,\\ University of L'Aquila, Italy}{stefano.leucci@univaq.it}{https://orcid.org/0000-0002-8848-7006}{}
\author{Alessandro Straziota}{Department of Enterprise Engineering, University of Rome ``Tor Vergata'', Italy}{alessandro.straziota@uniroma2.it}{https://orcid.org/0009-0008-4543-786X}{}

\authorrunning{D. Bilò, L. Gualà, S. Leucci, and A. Straziota} 
\Copyright{Davide Bilò,  Luciano Gualà, Stefano Leucci, and Alessandro Straziota} 



\hideLIPIcs
\nolinenumbers

\begin{document}

\maketitle

\begin{abstract}
A spanner is a sparse subgraph of a given graph $G$ which preserves distances, measured w.r.t.\ some distance metric, up to a multiplicative stretch factor.
This paper addresses the problem of constructing graph spanners w.r.t.\ the \emph{group Steiner metric}, which generalizes the recently introduced beer distance metric.
In such a metric we are given a collection of \emph{groups} of \emph{required} vertices, and we measure the distance between two vertices as the length of the shortest path between them  that traverses at least one required vertex from each group.

We discuss the relation between \emph{group Steiner spanners} and classic spanners and we show that they exhibit strong ties with \emph{sourcewise} spanners w.r.t.\ the shortest path metric.
Nevertheless, group Steiner spanners capture several interesting scenarios that are not encompassed by existing spanners. This happens, e.g., for the \emph{singleton case}, in which each group consists of a single required vertex, thus modeling the setting in which routes need to traverse certain points of interests (in any order).

We provide several constructions of group Steiner spanners for both the \emph{all-pairs} and \emph{single-source} case, which exhibit various size-stretch trade-offs. 
Notably, we provide spanners with almost-optimal trade-offs for the \emph{singleton} case.
Moreover, some of our spanners also yield novel trade-offs for classical sourcewise spanners.

Finally, we also investigate the query times that can be achieved when our spanners are turned into group Steiner distance oracles with the same size, stretch, and building time.
\end{abstract}

\section{Introduction}

Given an edge-weighted graph $G=(V,E)$, the distance between two vertices $s,t$ is typically measured
as the length of a shortest path having $s$ and $t$ as endvertices. Such \emph{shortest path metric} is pervasive in the study of optimization problems on graphs, yet there are natural scenarios that cannot be readily captured by such a metric. For example, consider the case in which a route from $s$ to $t$ needs to pass through at least one vertex from a distinguished set $R \subseteq V$ of \emph{required} vertices. These vertices might represent, e.g., grocery stores on your commute to work, charging stations when planning a trip in an electric vehicle, or special hosts when routing packets in a communication network.
The above metric was introduced in \cite{Bacic0S21} under the name \emph{beer-distance}.\footnote{In \cite{Bacic0S21}, the required nodes in $R$ correspond to breweries, and one seeks a shortest path from $s$ to $t$ among those that traverse at least one brewery.}
In particular, several works have been devoted to the problem of building a compact data structure that is able to quickly report (exact or approximate) beer distances between pair of vertices for special classes of graphs, such as outer planar graphs, interval graphs, or bounded tree-width graphs \cite{Bacic0S21,BacicMS23,Das0KMNW22,GudmundssonS23,HanakaOSS23}.
Such data structures are the beer distance rendition of \emph{distance oracles}, which are analogous data structures for the shortest path metric. 
Distance oracles and the closely related concept of \emph{graph spanners} have received a vast amount of attention in the area of graph sparsification. Informally, an $\alpha$-spanner of $G$ is a sparse subgraph $H$ of $G$ that preserves the all-pairs distances in $G$ up to a multiplicative \emph{stretch} factor of $\alpha$, and a distance oracle can be seen as data structure which allows for quick queries on the underlying spanner.\footnote{Actually, a distance oracle does not immediately imply the existence of a corresponding spanner, while a sparse spanner can always be thought as a compact distance oracle albeit with a large query time. Nevertheless, it is often the case that a spanner and its corresponding oracle are provided together. The main challenge in designing distance oracles lies in organizing the implicit distance information contained in the spanner in a way that allows for quick queries.}

\noindent The above discussion begs the following two natural questions:
\begin{itemize}
    \item What happens if paths are required to traverse more than one kind of required vertices? For example, in the commute from home to work one needs to visit both a grocery store \emph{and} a gas station, in some order. 
    \item What can be said about beer-distances for \emph{general} graph, i.e., when $G$ does not fall into one of the special classes of graphs mentioned above? 
\end{itemize}

Answering these questions is the focus of our paper, which will be devoted to designing spanners and distance oracles for general graphs and for the natural generalization of beer distance, which we name \emph{group Steiner distance}. 

Formally, given an undirected connected graph $G=(V, E)$ on $n$ vertices and with non-negative edge weights, and a collection of $k \ge 1$ (not necessarily disjoint) \emph{groups} $R_1, \dots, R_k \subseteq V$ of \emph{required} vertices, a \emph{group Steiner path} between two vertices $s$ and $t$ is a (not necessarily simple) path $\pi$ between $s$ and $t$ in $G$ such that $\pi$ includes at least one vertex from each group $R_i$. The \emph{group Steiner distance} between $s$ and $t$ is the length (w.r.t.\ the edge weights, with multiplicity) of the shortest group Steiner path between $s$ and $t$ (see \Cref{fig:example-group-steiner-path}).

Not surprisingly, as we discuss in~\Cref{app:complexity_gsp}, the problem of computing the group Steiner distance between two vertices is \NP-hard in general for large values of $k$.  On the other hand, the group Steiner distance between two vertices can be computed in \emph{Fixed Parameter Tractable (FPT)} time $2^k k n^{O(1)}$, which is polynomial when $k=O(\log n)$. Notice that it is easy to imagine scenarios in which $k$ is a small constant. 

The group Steiner distance coincides with the beer distance for $k=1$, but it also captures practical scenarios in which one wants to route entities through multiple points of interest, as in waypoint routing and other related motion planning problems \cite{ElbassioniFMS05,POP2024819,AmiriFS20}.

Using this novel metric, we can define a \emph{group Steiner $\alpha$-spanner} which is the analogous of $\alpha$-spanner when distances are measured w.r.t.\ the group Steiner distance metric. 
Group Steiner spanners exhibit interesting relations with classical graph spanners (for the shortest path metric). Indeed, one can observe that any shortest group Steiner path can be seen as the concatenation of up to $k+1$ sub-paths, each of which has $s$, $t$, or a vertex in $R = \cup_i R_i$ as endvertices, and is a shortest path in $G$ (see \Cref{sec:notation} for the details). Hence any \emph{sourcewise}\footnote{A sourcewise $R \times V$ $\alpha$-spanner of $G$ is a subgraph of $G$ that approximates all the $R \times V$ distances within a multiplicative factor of $\alpha$.} $R \times V$ $\alpha$-spanner for $G$ is also a group Steiner spanner with the same stretch factor. As a consequence, all the upper bounds on the sizes for classical sourcewise (and all-pairs) $\alpha$-spanners carry over to group Steiner spanners. 
It turns out that these two notions coincide when $R_1 = \dots = R_k=R$, which implies that, in general, it is not possible to obtain better size-stretch trade-offs than the ones for sourcewise spanners.
This is always the case for $k=1$.

\begin{figure}
    \centering
    \includegraphics[width=.4\linewidth]{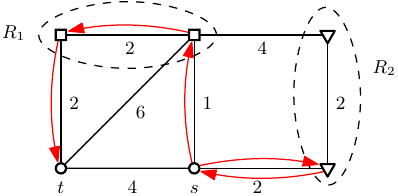}
    \caption{A shortest group Steiner path from $s$ to $t$ with length $9$. The required vertices of the two groups $R_1$ and $R_2$ are depicted as squares and triangles, respectively.}
    \label{fig:example-group-steiner-path}
\end{figure}

\subsection{Our results}

We mainly focus on group Steiner spanners and we show that their landscape exhibits a quite rich structure. In fact, while the problem of designing group Steiner spanners generalizes the problem of computing sourcewise spanners w.r.t.\ the shortest path metric, there are several interesting and natural classes of instances, like going from a source to a destination by passing through $k$ waypoints (in any order), for which the lower bounds for the sourcewise spanners do not apply. 

We start our investigation by considering exactly this scenario, that we call \emph{singleton case} because each group contains only one vertex. First, we pinpoint the extremal size-stretch trade-offs: on the one hand it is possible to build a group Steiner spanner with $O(kn)$ edges that preserves \emph{exact distances}; on the other hand $n-1$ edges are already sufficient to build a group Steiner \emph{tree spanner} with stretch $2$.\footnote{We observe that, in general, no sourcewise spanner w.r.t.\ the shortest path metric with stretch 2 and size $n-1$ exists. As an example, consider a complete unweighted $R\times (V\setminus R)$ bipartite graph and observe that none of the edges can be discarded if we are seeking for a sourcewise $R\times V$ $\alpha$-spanner, for any $\alpha <3$.} Both spanners can be constructed in polynomial-time. Moreover, we show that both results are tight, meaning that there are instances in which any group Steiner spanner preserving exact distances must contain $\Omega(kn)$ edges and instances for which any group Steiner spanner with stretch strictly less than $2$ must contain at least $n$ edges. 
We then consider intermediate stretch factors and show that $O(n/\varepsilon^2)$ edges suffice to build a group Steiner spanner having stretch  $\gamma + \varepsilon$ in polynomial time, where $\gamma$ is the approximation factor of a polynomial-time algorithm for the minimum-cost metric Hamiltonian path problem\footnote{See \Cref{app:complexity_gsp} for the formal definition of this problem.} and $\varepsilon \in (0,1)$. If one is willing to settle for an FPT building time w.r.t.\ $k$, then the above stretch can be improved to $1+\varepsilon$. These results are summarized in \Cref{tab:result1}~(a).

For general instances, we provide two \emph{recipes}, one of which can be thought as a generalization of Theorem 3 in \cite{ElkinFN17}. These recipes use existing $\alpha$-spanners w.r.t.\ the shortest-path metric to construct group Steiner spanners with stretch $2\alpha+1$ (see \Cref{tab:result1}~(b)). One recipe (\Cref{thm:all_pairs_general_2alpha_plus_1_spanner_bis}) upper bounds the size of the group Steiner spanner in terms of $n$ and $|R|$, while the other (\Cref{thm:all_pairs_general_2alpha_plus_1_spanner}) provides bounds w.r.t.\ the sizes of the groups $R_i$. Despite the simplicity of our constructions, by combining \Cref{thm:all_pairs_general_2alpha_plus_1_spanner_bis} with the spanners in \cite{AhmedBHKS21,BodwinWilliams,CyganGK13,ElkinGN23}, we obtain new trade-offs for sourcewise $R\times V$ spanners w.r.t.\ the shortest path metric. These results are marked with \ding{105} in \Cref{table:all_pairs_results_expanded}, which summarizes the current state of the art for group Steiner spanners.

We also consider the \emph{single-source} case\footnote{Roughly speaking, a single-source group Steiner spanner is only required to contain approximate group Steiner paths between a distinguished source vertex and all other vertices in $V$. See \Cref{sec:notation} for a formal definition.}, and we show that any $R \times R$ $\alpha$-spanner for the shortest path metric can be used to construct a single-source group Steiner spanner  with stretch $\alpha +1$. We also provide ad-hoc constructions achieving either the minimum stretch $\alpha=1$ or the minimum conceivable size $n-1$. These results are summarized in \Cref{tab:result1}~(c)), while the corresponding current state of the art for the single-source case is given in \Cref{table:single_source_weighted_results_expanded}. Finally, as an instrumental result to achieve our trade-offs, we also consider the \emph{single-pair} case and show that any group Steiner path can be sparsified to have $O(n)$ edges without increasing its length (see \Cref{crl:single_pair_1,crl:single_pair_2}).

\begin{table}
    \caption{Summary of our results. The size-stretch trade-offs of rows marked with \ding{73} are tight. The rightmost column reports the query time attained by a distance oracle with the same stretch, asymptotic size, and building time class (i.e., polynomial or FPT) of the associated spanner. 
    We use $\gamma$ to denote the approximation ratio of a polynomial-time algorithm for the minimum-cost metric Hamiltonian path problem, while $R=\bigcup_i R_i$ denotes the set of all required vertices.
    Expanded versions of (b) and (c) are respectively shown in \Cref{table:all_pairs_results_expanded} and in \Cref{table:single_source_weighted_results_expanded}.}
    \label{tab:result1}
    \begin{subtable}[h]{1\textwidth}
        \centering
        \caption{All-pairs, singleton case ($|R_i|=1 \; \forall i=1,\dots,k$).}
        \begin{tabular}{c|c|c|c|c||c|}
        \cline{2-6}
        & Stretch & Size & \!\! Building time \!\! & Reference & D.O. query time  \\
        \cline{2-6}
         \ding{73}\!\! & $1$ & $O(kn)$ & polynomial & \Cref{thm:singleton_preserver} & $O(2^k \cdot k^3)$ \\ 
         & $1+\varepsilon$ & $O(n/\varepsilon^2)$ & FPT & \Cref{thm:singleton_1_plus_eps} & $O(1/\varepsilon^2)$ \\
        & $\gamma+\varepsilon$ & $O(n/\varepsilon^2)$ & polynomial & \Cref{thm:singleton_gamma_plus_eps} & $O(1/\varepsilon^2)$ \\
        \ding{73}\!\! & $2$  & $n-1$ & polynomial & \Cref{thm:singleton_stretch_2} & $O(1)$ \\ \cline{2-6}
        \end{tabular}
    \end{subtable}
    \medskip

    \begin{subtable}[h]{1\textwidth}
        \centering
        \caption{All-pairs, general group sizes.}
        \begin{tabular}{|c|c|c|c||c|}
        \hline
        Stretch & Size & \!\! Building time \!\! & Reference & D.O. query time  \\
        \hline
        $2\alpha + 1$ & \!$kn + \! |\bigcup_i \left( R_i\times R_i  \; \alpha \text{-spanner} \right)|$\!  & polynomial & \Cref{thm:all_pairs_general_2alpha_plus_1_spanner} & $O(2^k k \cdot |R|^2+|R|^3)$ \\
        $2\alpha + 1$ & $n + |R\times R  \; \alpha \text{-spanner}|$ & polynomial & \Cref{thm:all_pairs_general_2alpha_plus_1_spanner_bis} & $O(2^k k \cdot |R|^2+|R|^3)$ \\
        \hline
        \end{tabular}
    \end{subtable}
    \medskip      
    
    \begin{subtable}[h]{1\textwidth}
        \centering
        \caption{Single-source, general group sizes.}
        \begin{tabular}{|c|c|c|c||c|}
        \hline
        Stretch & Size & \!\! Building time \!\! & Reference & D.O. query time \\
        \hline
        $1$ & $O(2^kn)$ & FPT & \Cref{thm:single_source_preserver} & $O(1)$ \\
        $3$ & $n-1$ & FPT & \Cref{thm:single_source_stretch_3} & $O(1)$ \\
        $\alpha + 1$ & $O(n) + | R \times R \; \alpha \text{-spanner}|$ & polynomial & \Cref{thm:single_source_general_alpha_plus_2_spanner} & $O(2^k k \cdot |R|^2+|R|^3)$  \\
        \hline
        \end{tabular}
    \end{subtable}
\end{table}

\begin{table}[t]
\caption{Known bounds for classical spanners for both weighted and unweighted graphs that yield the best trade-offs when used in our recipes of \Cref{tab:result1}~(b)~and~(c). \emph{Pairw.} denotes a pairwise spanner, i.e., a spanner which is only required to (approximately) preserve distances between pairs of vertices in $P \subseteq V^2$. Randomized constructions are marked with \epsdice{3}. A \emph{mixed} stretch of $(\alpha,\beta)$ means that the corresponding spanner $H$ approximates the distance from $s$ to $t$ in $G$ within a \emph{multiplicative} stretch of $\alpha$ plus an \emph{additive} stretch of $\beta$ times the maximum edge-weight, say $W_{s,t}$, along a shortest path from $s$ to $t$ in $G$, i.e., $d_H(s,t)\leq \alpha d_G(s,t)+\beta W_{s,t}$. Notice that a spanner with a mixed stretch of $(\alpha,\beta)$ is also a spanner with a purely multiplicative stretch of $\alpha+\beta$.}
\label{table:known_spanners_to_plug_in}
\begin{tabular}{|c|c|c|l|}
\multicolumn{4}{c}{\bf Weighted} \\ \hline 
Type              & Stretch & Size                      & \!Ref.\! \\ \hline
$V \times V$ & $2h-1$     & $O( n^{1  + \sfrac{1}{h}} )$  & \cite{AlthoferDDJS93} \\
Pairw. & $1$     & $O( n + |P|n^{\sfrac{1}{2}} )$  & \cite{CoppersmithE06} \\
Pairw. & $1$     & $O( n |P|^{\sfrac{1}{2}} )$  & \cite{CoppersmithE06} \\
Pairw. & \epsdice{3} $(1,2)$ & $O(n |P|^{\sfrac{1}{3}})$ &  \cite{AhmedBHKS21} \\
Pairw. & \epsdice{3} $(1,4 )$ & $O(n |P|^{\sfrac{2}{7}})$ &  \cite{AhmedBHKS21} \\
$R \times V$ & $4h - 1$ & \!\!\! $O( n \hspace{-0.3pt} + \hspace{-0.3pt} n^{\sfrac{1}{2}} |R|^{1 + \sfrac{1}{h}} )$ \!\!\! & \cite{ElkinFN17}\\
$R \times R$ & \!$(1, 2 \hspace{-0.3pt} + \hspace{-0.3pt} \varepsilon)$\! & $O(n |R|^{\sfrac{1}{2}}/\varepsilon)$ & \cite{ElkinGN23} \\
\hline
\end{tabular}
\begin{tabular}{|c|c|c|l|}
\multicolumn{4}{c}{\bf Unweighted} \\ \hline 
Type              & \!Stretch\! & Size                      & \!Ref.\! \\ \hline
Pairw. & $1$ & \!\!\! $O( n^{\sfrac{2}{3}} |P|^{\sfrac{2}{3}} + n |P|^{\sfrac{1}{3}} )$ \!\!\! & \cite{BodwinWilliams} \\
$R \times V$ & $(1,2)$ & $\widetilde{O}( n^{\sfrac{5}{4}} |R|^{\sfrac{1}{4}} )$ & \cite{Censor-HillelKP16} \\
$R \times V$ & $(1,4)$ & $\widetilde{O}( n^{\sfrac{11}{9}} |R|^{\sfrac{2}{9}} )$  & \cite{Kavitha17} \\
$R \times V$ & $(1,6)$ &  $\widetilde{O}( n^{\sfrac{6}{5}} |R|^{\sfrac{1}{5}} )$ & \cite{Kavitha17}\\
$R \times R$ & $(1,2)$ & $O( n |R|^{\sfrac{1}{2}} )$ & \cite{CyganGK13}\\
\hline
\end{tabular}
\end{table}

\begin{table}[t]
\caption{
State of the art for all-pairs group Steiner spanners of weighted graphs (left), along with additional bounds that only apply to unweighted graphs (right).
Randomized constructions are marked with \epsdice{3}. Combinations that are dominated by results with a better size-stretch trade-off are omitted. All building times are polynomial.
Results marked with \ding{105} are also novel sourcewise  $R \times V$ spanners w.r.t.\ the shortest path metric. The classical spanners used for our combinations are reported in \Cref{table:known_spanners_to_plug_in}.
}
\label{table:all_pairs_results_expanded}
\newcommand{\Rmax}{\ensuremath{R_{\text{max}}}}
\newcommand{\spd}{\phantom{\epsdice{3}}}
\newcommand{\thref}[1]{\hyperref[#1]{Th.~\!\getrefnumber{#1}}}
\newcommand{\obref}[1]{\phantom{0}\hyperref[#1]{Ob.~\!\getrefnumber{#1}}}
\hspace{-15pt}
\begin{tabular}{c|c|c|c|} 
    \multicolumn{1}{c}{} & \multicolumn{3}{c}{\bf All-pairs Weighted} \\ \cline{2-4}
    & $\alpha$ & Size & Reference \\ \cline{2-4}
    & \epsdice{3} $3$  & $O(n^{\sfrac{4}{3}} |R|^{\sfrac{1}{3}})$ & \obref{obs:all_pair_general_alpha}  +  \cite{AhmedBHKS21}\phantom{0}\\
    & \spd{} $3$  & $O(kn + k \Rmax^2 n^{\sfrac{1}{2}})$ & \thref{thm:all_pairs_general_2alpha_plus_1_spanner} + \cite{CoppersmithE06} \\
    & \spd{} $3$  & $O(n + |R|^2 n^{\sfrac{1}{2}})$ & \obref{obs:all_pair_general_alpha} + \cite{ElkinFN17} \\
    & \spd{} $3$ & $O( n^{\sfrac{3}{2}} )$ & \obref{obs:all_pair_general_alpha} + \cite{AlthoferDDJS93}\phantom{0} \\
    & \epsdice{3} $5$  & $O(n^{\sfrac{9}{7}} |R|^{\sfrac{2}{7}})$ & \obref{obs:all_pair_general_alpha} +  \cite{AhmedBHKS21}\phantom{0} \\
    & \spd{} $5$ & $O( n^{\sfrac{4}{3}} )$  & \obref{obs:all_pair_general_alpha} + \cite{AlthoferDDJS93}\phantom{0} \\
    \ding{105}\!\! & \epsdice{3} $7$ & $O(n |R|^{\sfrac{2}{3}})$ & \thref{thm:all_pairs_general_2alpha_plus_1_spanner_bis} + \cite{AhmedBHKS21}\phantom{0}  \\    
    & \spd{} $7$  & $O(n + |R|^{\sfrac{3}{2}} n^{\sfrac{1}{2}})$ & \obref{obs:all_pair_general_alpha} + \cite{ElkinFN17} \\
    & \spd{} $7$ & $O( n^{\sfrac{5}{4}} )$  & \obref{obs:all_pair_general_alpha} + \cite{AlthoferDDJS93}\phantom{0} \\
    \ding{105}\!\! & $7+\varepsilon$ & $O(n |R|^{\sfrac{1}{2}}/\varepsilon)$ & \thref{thm:all_pairs_general_2alpha_plus_1_spanner_bis} + \cite{ElkinGN23} \\
    & $2h-1$ & $O( n^{1+\sfrac{1}{h}} )$  & \obref{obs:all_pair_general_alpha} + \cite{AlthoferDDJS93}\phantom{0} \\
    & $4h-1$ & $O(n + n^{\sfrac{1}{2}}|R|^{1+\sfrac{1}{h}})$ & \obref{obs:all_pair_general_alpha} + \cite{ElkinFN17} \\ \cline{2-4}
\end{tabular}
\hspace{2pt}
\begin{tabular}{c|c|c|c|}
    \multicolumn{1}{c}{} & \multicolumn{3}{c}{\bf All-pairs Unweighted} \\ \cline{2-4} 
    & $\alpha$ & Size & Reference \\ \cline{2-4}
    \ding{105}\!\! & $3$ & $O( n^{\sfrac{2}{3}} |R|^{\sfrac{4}{3}} + n |R|^{\sfrac{2}{3}} )$ & \thref{thm:all_pairs_general_2alpha_plus_1_spanner_bis} + \cite{BodwinWilliams}\phantom{0} \\
    & $3$ & \!\!\! $O(k n^{\sfrac{2}{3}} \Rmax^{\sfrac{4}{3}} \!\! + \! k n \Rmax^{\sfrac{2}{3}})$ \!\!\! & \thref{thm:all_pairs_general_2alpha_plus_1_spanner} + \cite{BodwinWilliams}\phantom{0} \\ 
    & $3$ & $\widetilde{O}(n^{\sfrac{5}{4}}|R|^{\sfrac{1}{4}})$ & \obref{obs:all_pair_general_alpha} + \cite{Censor-HillelKP16}\phantom{0}\\
    & $5$ & $\widetilde{O}(n^{\sfrac{11}{9}}|R|^{\sfrac{2}{9}})$ & \obref{obs:all_pair_general_alpha} + \cite{Kavitha17}\\
    & $7$ & $\widetilde{O}(n^{\sfrac{6}{5}}|R|^{\sfrac{1}{5}})$ & \obref{obs:all_pair_general_alpha} + \cite{Kavitha17}\\ 
    & $7$  & $O(n + |R|^{\sfrac{3}{2}} n^{\sfrac{1}{2}})$ & \obref{obs:all_pair_general_alpha} + \cite{ElkinFN17} \\
    \ding{105}\!\! & $7$ & $O(n |R|^{\sfrac{1}{2}})$ & \thref{thm:all_pairs_general_2alpha_plus_1_spanner_bis} + \cite{CyganGK13} \\ \cline{2-4}
\end{tabular}
\end{table}

\begin{table}[t]
\caption{State of the art for single-source group Steiner spanners of weighted graphs (top), along with additional bounds that only apply to unweighted graphs (bottom).
Randomized constructions are marked with \epsdice{3}. Combinations that are dominated by results with a better size-stretch trade-off are omitted. The classical spanners used for our combinations are reported in \Cref{table:known_spanners_to_plug_in}.} 
\label{table:single_source_weighted_results_expanded}
\newcommand{\spd}{\phantom{\epsdice{3}}}
\newcommand{\thref}[1]{\hyperref[#1]{Th.~\getrefnumber{#1}}}
\newcommand{\obref}[1]{\phantom{0}\hyperref[#1]{Ob.~\!\getrefnumber{#1}}}
\centering
\begin{tabular}{|c|c|c|c|}
    \multicolumn{4}{c}{\bf Single-source weighted} \\ \hline
    $\alpha$ & Size & Building time & Reference \\ \hline
    \spd{} $1$       & $O(2^k n)$ &  FPT & \Cref{thm:single_source_preserver} \\
    \spd{} $2$       & $O( n + |R|^2n^{\sfrac{1}{2}} )$ & polynomial & \thref{thm:single_source_general_alpha_plus_2_spanner} + \cite{CoppersmithE06}  \\
    \spd{} $3$       & $n-1$ &  FPT & \Cref{thm:single_source_stretch_3} \\
    \spd{} $3$ & $O( n^{\sfrac{3}{2}} )$  & polynomial & \obref{obs:all_pair_general_alpha} + \cite{AlthoferDDJS93}  \\
    \epsdice{3} $4$    & $O(n |R|^{\sfrac{2}{3}})$  & polynomial & \thref{thm:single_source_general_alpha_plus_2_spanner} + \cite{AhmedBHKS21}\phantom{0}  \\
    $4+\varepsilon$ & $O(n |R|^{\sfrac{1}{2}} / \varepsilon)$ & polynomial & \thref{thm:single_source_general_alpha_plus_2_spanner} + \cite{ElkinGN23} \\
    $2h-1$ & $O( n^{1+\sfrac{1}{h}} )$  & polynomial & \obref{obs:all_pair_general_alpha} + \cite{AlthoferDDJS93}\phantom{0} \\
    $4h-1$ & $O(n+ n^{\sfrac{1}{2}} |R|^{1+\sfrac{1}{h}})$ & polynomial & \obref{obs:all_pair_general_alpha} + \cite{ElkinFN17} \\ \hline
\end{tabular}
\medskip

\begin{tabular}{|c|c|c|c|}
    \multicolumn{4}{c}{\bf Single-source unweighted} \\ \hline
    $\alpha$ & Size & Building time & Reference  \\ \hline
    $2$ & $O(n^{\sfrac{2}{3}} |R|^{\sfrac{4}{3}} + n |R|^{\sfrac{2}{3}} )$ & polynomial & \thref{thm:single_source_general_alpha_plus_2_spanner} + \cite{BodwinWilliams}\phantom{0}  \\
    $4$ & $O(n |R|^{\sfrac{1}{2}})$ & polynomial & \thref{thm:single_source_general_alpha_plus_2_spanner} + \cite{CyganGK13} \\
    $6$ & $\widetilde{O}( n^{\sfrac{11}{9}} |R|^{\sfrac{2}{9}} )$ & polynomial & \thref{thm:single_source_general_alpha_plus_2_spanner} + \cite{Kavitha17}  \\
    $8$ &  $\widetilde{O}( n^{\sfrac{6}{5}} |R|^{\sfrac{1}{5}} )$ & polynomial & \thref{thm:single_source_general_alpha_plus_2_spanner} + \cite{Kavitha17}    \\
    \hline
\end{tabular}
\end{table}

\subparagraph*{Turning our spanners into group Steiner distance oracles.}

We also investigate the problem of turning our group Steiner spanners into group Steiner distance oracles. For each of our spanners, we provide a corresponding oracle with the same stretch, the same asymptotic size and the same class of building time.\footnote{We classify the building times into one of two coarse classes, namely \emph{polynomial} and \emph{FPT}, depending on whether the spanner/oracle can be computed in time $n^{O(1)}$ or $f(k)n^{O(1)}$.} 
The distance query times are reported in \Cref{tab:result1}. Some of these query times are constant, and in this case our oracles are also able to report a corresponding group Steiner path in an additional time proportional to the number path's edges.  
The remaining query times are exponential in $k$ and this is unavoidable. Indeed, consider the group Steiner spanner in \Cref{tab:result1}~(a) with stretch $1$ and polynomial building time. It can be shown (see \Cref{app:complexity_gsp} for details) that any corresponding oracle with polynomial query time would be able to report the cost of a minimum-cost metric Hamiltonian path which is known to be \NP-hard. For similar reasons, any oracle for general group sizes that has polynomial  building time and stretch $\log^{2-\varepsilon} k$, for constant $\varepsilon>0$, cannot have polynomial query time (regardless of its size), even for the single-pair oracle, due to the inaproximability of computing group Steiner distances. 

Finally, we emphasize that the two oracles for the singleton case with stretch $1+\varepsilon$ and $\gamma+\varepsilon$ are, in some sense, tight. 
Indeed, since a distance oracle can be used to compute group Steiner distances, the building time of the former oracle (having stretch $1+\varepsilon$) cannot be improved to polynomial time since the minimum-cost metric Hamiltonian path problem is $\mathsf{APX}$-hard \cite{KarpinskiLS15}; while improving the stretch to a value better than $\gamma$ in the latter oracle would provide a better than $\gamma$-approximation for the minimum-cost metric Hamiltonian path problem.

\subsection{Related work}

There is a huge body of literature on graph spanners and distance oracles w.r.t.\ the shortest-path metric. 
Since we mostly focus on the stretch-size trade-offs of our group Steiner spanners, in the following we discuss the related work providing the best size-stretch trade-offs for spanners. The reader interested in efficient computation of spanners is referred to~\cite{LeS22} and to the references therein.

A classical result shows that it is possible to build all-pairs spanners with stretch $2h-1$ and size $O(n^{1+\frac{1}{h}})$, for every integer $h \ge 1$~\cite{AlthoferDDJS93} . For $h \in \{1,2,3,5\}$ these asymptotic bounds are (unconditionally) tight~\cite{Tits1959SurLT,Wenger1991}, and in general, for every $h$, matching asymptotic lower bounds can be proved assuming the Erd\H{o}s girth conjecture~\cite{erdos1965some}.
The reader is referred to \cite{BodwinSpanners} for a survey which also discusses other notions of stretch (e.g., additive and mixed distortions), as well as generalizations in which good distance approximations only need to be maintained between specific pairs of vertices of interest (as sourcewise, subsetwise, or pairwise spanners). 
In~\cite{ThorupDO,Chechik14,Chechik15}, the authors show how to build, in polynomial time, distance oracles achieving the same size-stretch trade-offs and having constant query time.

As mentioned above, group Steiner paths for the special case $k=1$ are known in the literature as \emph{beer paths}. The notion of beer paths (and the corresponding \emph{beer distance}) has been first introduced in \cite{Bacic0S21,BacicMS23}, where the authors show how to construct \emph{beer distance oracles} for outerplanar graphs that are able to report exact beer distances. Subsequent works showed how to construct beer distance oracles for interval graphs \cite{Das0KMNW22} and graphs with bounded treewidth \cite{GudmundssonS23}. Construction of beer distance oracles for graphs that admit either good tree decomposition or good triconnected component decompositions have been studied in \cite{HanakaOSS23}.
None of the above results yields a non-trivial beer distance oracle for general graphs, hence they cannot readily be compared with group Steiner distance oracles (which also handle $k > 1$ groups).

Several distance metrics involving paths that are required to traverse groups of vertices have already been considered in the context of optimization problems. For example, the \emph{generalized TSP} problem \cite{POP2024819} asks to find a shortest tour in a graph that visits at least one vertex from each group, which corresponds to finding the shortest  group Steiner path from a required vertex in the optimal tour to itself. Elbassoni et al., studied a geometric version of the above problem called \emph{Euclidean group TSP} \cite{ElbassioniFMS05}.
A related optimization problem involving both waypoints and capacity constraints is known as \emph{waypoint routing} and has been studied in \cite{AmiriFS20}. 

Finally, we point out that our metric should not be confused with a different measure also called \emph{Steiner distance} that is defined as the weight of the lightest Steiner tree connecting as set of vertices (see, e.g., \cite{Chartrand1989,mao2017steiner}).

\section{Preliminaries}
\subsection{Notation}\label{sec:notation}

We denote by $G=(V,E)$ a connected undirected graph with $n$ vertices, $m$ edges, and with a non-negative edge-weight $w(e)$ associated with each $e \in E$, We also denote by $\mathcal{R} = \{R_1, \dots, R_k\}$ a collection of $k$ non-empty subsets of $V$, which we refer to as \emph{groups} of \emph{required} vertices. We denote by $R = \bigcup_{i=1}^k R_i$. 

Throughout this work, we use the term path to refer to walks, i.e., our paths are not necessarily simple. 
A \emph{group Steiner path} from $s$ to $t$ in $G$ w.r.t.\ $\mathcal{R}$ is a path from $s$ to $t$ that contains at least one vertex from each group. The length $w(\pi)$ of a path $\pi$ is the sum of all its edge weights, with multiplicity (see \Cref{fig:example-group-steiner-path}).
The \emph{group Steiner distance} $\BD_G(s,t \mid \mathcal{R})$ between $s$ and $t$ w.r.t.\ $\mathcal{R}$ in $G$ is the length of the shortest group Steiner path from $s$ to $t$ in $G$. Whenever $\mathcal{R}$ is clear from the context, we omit it from the notation.

Given $\alpha \geq 1$, a \emph{group Steiner $\alpha$-spanner of $G$} w.r.t.\ $\mathcal{R}$ is a spanning subgraph $H$ of $G$ such that:
\begin{equation}
    \label{eq:group_Steiner_spanner}
    \BD_H(s,t) \le \alpha \cdot \BD_G(s,t),
\end{equation} for every $s, t \in V$. 
We denote by $|H|$ the \emph{size} of $H$ which corresponds to the number of edges contained in $H$.
When $\alpha=1$, a group Steiner $1$-spanner of $G$ is called a \emph{group Steiner preserver} as it preserves group Steiner distances between all-pairs of vertices. 

The classical notion of graph $\alpha$-spanner w.r.t.\ the shortest-path metric is analogous, once the group Steiner distances $\BD_H(s,t)$ and $\BD_G(s,t)$ in \Cref{eq:group_Steiner_spanner} are replaced with the lengths $d_H(s,t)$ and $d_G(s,t)$ of a shortest path in $H$ and $G$, respectively.

Regardless of the distance metric of interest, we can restrict the pairs of vertices for which the distances in $H$ must $\alpha$-approximate the corresponding distances in $G$ to those in the set $S \times T$, for some choice of $S,T \subseteq V$. The definition of $\alpha$-spanner in~\Cref{eq:group_Steiner_spanner} corresponds to the \emph{all-pairs} case in which $S,T=V$. In \emph{sourcewise} spanners we have $S \subseteq V$ and $T=V$. \emph{Single-source} spanners are a special case of sourcewise spanners in which $S=\{s\}$, for some source vertex $s \in V$. Finally, in \emph{subsetwise} spanners we have $S=T$, with $S\subseteq V$.

Given a path $\pi$ from $s$ to $v$ in $G$ and a path $\pi'$ from $v$ to $t$ in $G$, we denote by $\pi \circ \pi'$ the path obtained by concatenating $\pi$ with $\pi'$. Moreover, given two occurrences $x,y$ of vertices in $\pi$, we denote by $\pi[x:y]$ the subpath of $\pi$ between $x$ and $y$.
Notice that the same vertex can appear multiple times in $\pi$ however, whenever the occurrence of interest is clear from context, we may slightly abuse the notation and use vertices in place of their specific occurrences.

We conclude this preliminary section by observing a  structural property of group Steiner paths and its important immediate consequences. The property can be proved by a simple cut-and-paste argument, and intuitively shows that a shortest group Steiner path can be seen as a concatenation of up to $k+1$ subpaths, each of which is a shortest path in $G$ between two vertices in $R \cup \{s,t\}$. This is formalized in the following:

\begin{lemma}    
    \label{lemma:group_steiner_path_decomposition}
    Let $\pi = \langle s= v_0, v_1, \dots, v_\ell = t \rangle$ be a shortest group Steiner path between two vertices $s$ and $t$ in $G$. 
    Let $j_1, \dots, j_h$ be $h$ indices such that $0 \le j_1 < j_2 < \dots < j_h \le \ell$ and $\{ v_{j_1}, \dots, v_{j_h} \} \cap R_i \neq \emptyset$ for all $i = 1 \dots, k$. 
    For every $i=0, \dots, h$, $\langle v_{j_i}, v_{j_i+1}, \dots, v_{j_{i+1}} \rangle$ is a shortest path between $v_{j_i}$ and $v_{j_{i+1}}$ in $G$, where  $j_0 = 0$ and $j_{h+1} = \ell$.
\end{lemma}

The above property immediately implies the following results:

\begin{observation}
    \label{obs:all_pair_general_alpha}
    Any sourcewise $R \times V$ (and hence also any all-pairs $V \times V$) $\alpha$-spanner w.r.t.\ the shortest path metric is an all-pairs group Steiner $\alpha$-spanner.
\end{observation}
\begin{proof}
Consider an optimal group Steiner path $\pi^*$ from $s$ to $t$. Clearly, $\pi^*$ must traverse all the groups in $\mathcal{R}$.
Without loss of generality, assume that $\pi^*$ traverses the groups $R_1,\dots,R_k$ in this order (otherwise, re-index the groups properly), and let $r_i$ be the first required vertex in $R_i$ reached by $\pi^*$.
Then, using \Cref{lemma:group_steiner_path_decomposition}:
\begin{align*}
    \BD_H(s,t)
    &\leq d_H(s,r_1) + d_H(r_1,r_2) + \dots + d_H(r_k,t)\\
    &\le \alpha \cdot (d_G(s,r_1) + d_G(r_1,r_2) + \dots + d_G(r_k,t)) \le \alpha  \cdot \BD_G(s,t). \tag*{\qedhere}
\end{align*}
\end{proof}

\begin{observation}
    \label{obs:gs_spanner_eq_source_wise}
    Any group Steiner $\alpha$-spanner with $R_1 = R_2 = \dots = R_k = R$ is a sourcewise $R \times V$ $\alpha$-spanner w.r.t.\ the shortest path metric.
\end{observation}
\begin{proof}
Let $H$ be a group Steiner spanner with $R_1 = R_2 = \dots = R_k = R$, and pick $r \in R$. For any $t \in V$, \Cref{lemma:group_steiner_path_decomposition} implies that $\BD_G(r, t) = d_G(r,t)$, hence: $d_H(r,t) \le \BD_H(r, t) \le \alpha \cdot \BD_G(r,t) = \alpha \cdot d_G(r,t)$.
\end{proof}

\subsection{Single-pair group Steiner spanners}

Our technical discussion begins with a result for the single-pair case that, other than being interesting in its own regard, will also be instrumental to construct our group Steiner spanners.

\begin{lemma}\label{lm:linear_size_path}
Let $\pi$ be a group Steiner path from $s$ to $t$ in $G$.
We can process $\pi$ in polynomial time to build a group Steiner path $\pi'$ from $s$ to $t$ in $G$ that traverses at most $2(n-1)$ edges (with multiplicity) and such that $w(\pi') \leq w(\pi)$. 
\end{lemma}
\begin{proof}
We create a \emph{multi}-graph $H$ with vertex-set $V(\pi)$ by i) following $\pi$ and adding each traversed edge to $H$ (notice that this results in parallel edges whenever the same edge is traversed more than once in $\pi$), and ii) adding a \emph{special} edge $(s,t)$ of weight $0$.
Notice, that, by construction, $w(H)=w(\pi)$; moreover, $H$ is Eulerian.

Compute any spanning tree $T$ of $H$ that contains the special edge $(s,t)$, let $D \subseteq V(T)$ be the set of vertices of $T$ with \emph{odd} degree in $T$, and observe that $|D|$ is even.
Since $H$ is Eulerian, all vertices of $H$ have even degree, hence the vertices of odd degree in $H'=(V(H), E(H) \setminus E(T))$ are exactly those in $D$.

We will show how to partition $D$ into pairs of vertices $\{ x_1, y_1 \}, \ldots, \{ x_{|D|/2}, y_{|D|/2} \}$ and find a path $\pi_i$ between $x_i$ and $y_i$ in $H'$ for each such pair.
All the paths $\pi_i$ will be edge-disjoint and will use at most $n-1$ edges in total. Therefore, the multi-graph $H^*$ obtained by adding to $T$ all the edges in $\bigcup_{i=1}^{|D|/2} E(\pi_i)$ is Eulerian and it is a subgraph of $H$.
We can then find the group Steiner path $\pi'$ by computing an Eulerian tour of $H^*$ and removing the special edge $(s,t)$.
Clearly, $w(\pi') \le w(H^*)\le w(H)= w(\pi)$.
Moreover, the number of edges of $\pi'$ is at most $2(n-1)$, since $H^*$ is the union of the tree $T$ with at most $n-1$ additional edges.

It remains to show how the partition of $D$ and the corresponding paths can be found.
In order to do so, we consider a spanning forest $F$ of $H'$ and we separately handle each connected component $T'$ of $F$ that contains at least one vertex in $D$.

Consider any such $T'$ and notice that it must actually contain an even number of vertices $D' = V(T') \cap D$. Then, we can match the vertices of $D'$ in pairs in such a way that the (unique) paths in $T'$ between matched vertices are edge-disjoint (see \Cref{fig:edge_disjoint_paths}).

\begin{figure}
    \centering
    \includegraphics[width=.4\linewidth]{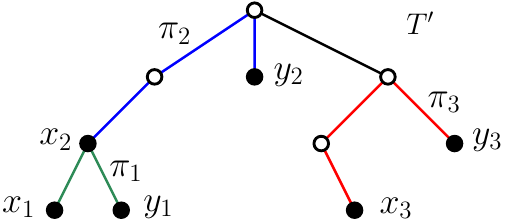}
    \caption{The subtree $T'$ used in the proof of \Cref{lm:linear_size_path}. The vertices in $D'$ are black and are connected in pairs by the edge-disjoint paths $\pi_1, \dots, \pi_{|D'| / 2}$, where the endvertices of $\pi_i$ are $x_i$~and~$y_i$.}
    \label{fig:edge_disjoint_paths}
\end{figure}

To match the vertices in $D'$ we use the following iterative algorithm: initially all vertices in $D'$ are unmarked, we find a pair of unmarked vertices $x,y \in D'$ with the deepest lowest common ancestor in $T'$, we match $x$ with $y$ and mark both vertices $x$ and $y$, and we repeat the above procedure until all vertices in $D'$ are marked. Notice that, since all paths between matched vertices, are edge-disjoint and belong to $F$, the overall number of edges is at most $n-1$. 
\end{proof}

From~\Cref{lm:linear_size_path}, we can easily derive the following corollaries. The first corollary is a direct consequence that a shortest group Steiner path can be computed in $2^k k \cdot n^{O(1)}$ time (see \Cref{sec:computing_group_Steiner_paths}), while the second corollary, which holds only for the singleton case, comes from the fact that we can compute a $3/2$-approximation of the shortest group Steiner path in $G$ (see \Cref{apx:steiner_path_singleton_case}). 

\begin{corollary}\label{crl:single_pair_1}
We can compute a single-pair shortest group Steiner path of size at most $2(n-1)$ in $2^k k \cdot n^{O(1)}$ time.
\end{corollary}

\begin{corollary}\label{crl:single_pair_2}
For the singleton case in which $|R_i|=1$ for all $i=1,\ldots,k$, there is a polynomial-time algorithm that computes a single-pair group Steiner path with stretch $3/2$ and size at most $2(n-1)$. 
\end{corollary}

\section{Group Steiner spanners in the singleton case}

In this section we consider the special case in which each group $R_i$ in $\mathcal{R}$ contains a single vertex $r_i$, hence $R = \{r_1, r_2, \dots, r_k\}$.

\subsection{A group Steiner preserver}\label{sec:preserver_singleton}

\Cref{lemma:group_steiner_path_decomposition} implies that union of $k$ shortest-path trees $T_1, \dots, T_k$ of $G$, where $T_i$ is rooted in $r_i$, is a group Steiner preserver of size $O(kn)$.
The size of such a preserver is asymptotically optimal, even when $G$ is unweighted.
To see this, let $k \ge 3$ and consider a graph $G$ consisting of a cycle $\langle r_1, r_2, \dots, r_k, r_1 \rangle$ on the $k$ required vertices, along with $n-k$ additional vertices $v_1, \dots, v_{n-k}$ (see \Cref{fig:lb_singleton}).\footnote{For $k \in {1, 2}$, a trivial lower bound of $\Omega(n) = \Omega(kn)$ clearly holds.} All vertices $v_i$ have $\lfloor k / 3 \rfloor$ incident edges, where the $j$-th such edge is $e_{i,j} = (v_i, r_{1 + 3(j-1)})$. 
Given any $v_i$ and $r_{1+3(j-1)}$, there exists a unique simple path from $v_i$ to $r_{2+3(j-1)}$ spanning $v_i$ and $\{r_1, \dots, r_k\}$, and such a path uses the edge $e_{i,j}$. As a consequence, any group Steiner preserver $H$ must contain all edges $e_{i,j}$, which implies that $H$ must have size $\Omega(kn)$.

\begin{theorem}
    \label{thm:singleton_preserver}
    In the singleton case, it is possible to compute a group Steiner preserver of size $O(kn)$ in polynomial time. Moreover, there are unweighted graphs $G$ such that any group Steiner preserver of $G$ has size $\Omega(kn)$.
\end{theorem}

\subsubsection*{A corresponding distance oracle}

The above group Steiner preserver can be turned into a group Steiner distance oracle with the same asymptotic size, stretch and building time. In order to do this we maintain an additional complete auxiliary graph $H_R$ on the vertices in $R$.
The weight of a generic edge $(u,v)$ in $H_R$ is $d_G(u,v)$. Notice that $H_R$ can be computed in polynomial time and its size is $O(k^2)=O(kn)$. 
To report $\BD(s,t)$, we augment $H_R$ by adding (i) two new vertices $s'$ and $t'$ along with all edges $(s', r)$ and $(t', r)$ with weight $d_G(s,r)$ and $d_G(t,r)$, respectively, for $r \in R$. Then, we compute a shortest group Steiner path $\pi$ between $s'$ and $t'$ in $H_R$ using the algorithm of \Cref{sec:computing_group_Steiner_paths}, which requires time 
$O(2^k k^3 + k(k^2 + k \log k))=O(2^k k^3)$ since $|R|=k$ and the number of vertices in $H_R$ is $k+2$. We can also report a corresponding group Steiner path in $G$ by replacing each edge of $\pi$ with a shortest path in $G$ between the edges' endvertices (which are stored in the $T_i$'s). This can be done in time proportional to the number of edges of the reported path.

\subsection{A spanner with stretch \texorpdfstring{$1+\varepsilon$}{1+epsilon} and \texorpdfstring{$O(n/\varepsilon^2)$}{O(n/epsilon\^2)} edges}

By the lower bound in~\Cref{thm:singleton_preserver}, any group Steiner $\alpha$-spanner of size $o(kn)$ must have a stretch of $\alpha > 1$. In this section we present an algorithm that builds a group Steiner spanner of linear size and stretch $1+\varepsilon$ in time $2^k k n^{O(1)}$, for every constant $\varepsilon > 0$.
We then show how to reduce the building time of our group Steiner spanner to polynomial at the cost of increasing the stretch to $\gamma+\varepsilon$ by using a $\gamma$-approximation algorithm for the minimum-cost metric Hamiltonian path problem as a black-box.
By keeping the same trade-offs among size, stretch, and building time, we show how to convert both spanners to group Steiner distance oracles with query time $O(1/\varepsilon^2)$.  The pseudocode for constructing our $(1+\varepsilon)$ group Steiner spanner can be found in \Cref{alg:1_plus_eps}.

\begin{algorithm}[!t]
\caption{Our algorithm for computing a group Steiner $(1+\epsilon)$-spanner of a graph $G$ in the singleton case. $R$ denotes the set of required vertices.}
\label{alg:1_plus_eps}

\SetKwFunction{Micro}{PartitionIntoMicroTrees}
\SetKwFunction{Cluster}{Cluster}
\SetKwProg{Fn}{Function}{:}{\relax}

\tcc{Returns a rooted forest $F$ of $G$ with one tree $T_c$ for each $c \in C$. If vertex $v$ is in $T_c$ then $d_{T_c}(v,c) = d_G(v, c) \le d_G(v, c')$ for   $c' \in C$.}
\Fn{\Cluster{$C$}}{
    $G' \gets $ graph obtained from $G$ by adding a new vertex $s^*$ and all edges $(s^*, c)$ of weight $0$ for $c \in C$\;
    $\widetilde{T} \gets $ a shortest path tree of $G'$ from $s^*$ that contains all edges incident to $s^*$\;
    $F \gets $ the forest obtained from $\widetilde{T}$ by deleting $s^*$ and all its incident edges\;
    \Return $F$\;
}
\BlankLine
\tcc{Returns a partition of a group Steiner tree $T$ of $G$ w.r.t.\ $R$ into a collection $\mathcal{T}$ containing $O(\frac{1}{\epsilon})$ edge-disjoint \emph{micro-trees}.}
\Fn{\Micro{$T$}}{
    $W \gets \frac{\varepsilon}{4} w(T)$\;
    $\mathcal{T} \gets \emptyset$\tcp*{A collection of micro-trees}
    $T' \gets$ tree obtained by rooting $T$ in an arbitrary vertex\;
    \While{$w(T') > W$}{
        $v \gets $ deepest node in $T'$ such that $w(T_v) \ge W$\;
        $u_1, \dots, u_k \gets$ children of $v$ in $T'$, in an arbitrary order\;
        $i \gets$ smallest index such that $\sum_{j=1}^i \big( w(T'_{u_j}) + w(u_j, v) \big) > W$\;
        $T'' \gets $ subtree of $T'$ induced by $v$ and all the vertices in $T_{u_j}$ for $j<i$\;
        $\mathcal{T} \gets \mathcal{T} \cup \{T'', T'_{u_i}\}$\;
        $T' \gets $ tree obtained from $T'$ by deleting all the vertices in $T_{u_j}$ for $j \le i$\;
    }
    \Return $\mathcal{T} \cup \{ T' \}$\;
    }
\BlankLine
$T \gets$ \hbox{a Steiner tree w.r.t.\ $R$ such that~$w(T)$ is at most twice the weight of an} optimal Steiner tree\;
${\mathcal{T}_1, \dots, \mathcal{T}_h} \gets$ \Micro($T$)\;
$\overline{G} \gets $ complete graph with vertex set $R$, the weight of a generic edge $(u,v)$ is $d_G(u,v)$\;
\For{$i \gets 1, \dots, h$}
{
    \For{$j \gets i, \dots, h$}
    {
        $\pi'_{i,j} \gets$ minimum cost path among all Hamiltonian paths in $\overline{G}$  having an endvertex in $\mathcal{T}_i$ and the other endvertex in $\mathcal{T}_j$\;
        $\pi''_{i,j} \gets$ path obtained from $\pi$ by replacing each edge $(u,v)$ with a shortest path from $u$ to $v$ in $G$ (w.r.t.\ the shortest path metric)\;
        $\pi_{i,j} \gets$ sparsify $\pi''_{i,j}$ as shown in \Cref{lm:linear_size_path}\;
    }
}
\BlankLine
\For{$i \gets 1, \dots, h$}
{
    $F_i \gets$ \Cluster{$V(T_i)$}, where $V(T_i)$ denotes the set of vertices in $T_i$\;
}
\BlankLine
\Return $H= T \cup \left( \bigcup_{i=1}^h \bigcup_{j=i}^h \pi_{i,j} \right) \cup \left( \bigcup_{i=1}^h F_i \right)$\;
\end{algorithm}

We start by defining an auxiliary \emph{clustering procedure} that will be useful for describing all our spanner constructions. Given a set of \emph{centers} $C \subseteq V$, the procedure computes a spanning forest $F$ of $G$ with $|C|$ rooted trees with the following properties: (i) the root of each tree is a distinct vertex in $C$, (ii) the unique path in $F$ from a vertex $v \in V$ to the root $c$ of its tree is a shortest path between $v$ and $c$ in $G$, (iii) $c$ is (one of) the closest center(s) to $v$. The procedure first constructs a graph $G'$ which is obtained from $G$ by adding a \emph{dummy source vertex} $s^*$ along a \emph{dummy edge} $(s^*, c)$ of weight $0$ for each $c \in C$. Then, it computes a shortest-path tree $\widetilde{T}$ of $G'$ from $s^*$ that contains all the dummy edges. Finally, it returns the forest $F$ obtained by deleting $s^*$ (and all its incident dummy edges) from $\widetilde{T}$.

We show how to compute our group Steiner spanner $H$. Let $T$ be a Steiner tree of $G$ w.r.t.\ the required vertices $R$ whose total weight $w(T)$ is at most twice that of the optimal Steiner tree $T^*$.
It is known that $T$ can be computed in polynomial time~\cite{Vazirani}.

We subdivide $T$ into $O(\frac{1}{\varepsilon})$ edge-disjoint \emph{micro-trees} $\mathcal{T}_1, \dots, \mathcal{T}_h$ that altogether span all vertices of $T$ and such that each micro-tree is a subtree of $T$ of weight at most $W=\frac{\varepsilon}{4} w(T)$. 

The subdivision is computed by the following iterative procedure, which keeps track of the part $T'$ yet to be divided. Initially, $T'$ is obtained by rooting $T$ in an arbitrary vertex.
As long as $T'$ has weight larger than $W$, we find a node $v$ such that the weight of the subtree  $T'_v$ of $T'$ rooted in $v$ is larger than $W$ and the depth of $v$ in $T'$ is maximized.
Let $u_1, u_2, \dots$ be the children of $v$ in $T'$, and let $i$ be the smallest index such that $\sum_{j=1}^i \big( w(T'_{u_j}) + w(u_j, v) \big) > W$. We create two micro-trees: one consists of the subtree of $T'_v$ induced by $v$ and all the vertices in $T'_{u_1}, \dots, T'_{u_{i-1}}$ (if any), and the other consists of $T'_{u_i}$. Notice that our choice of $v$ and $i$ ensures that both micro-trees have weight at most $W$.
Finally, we delete all vertices in $T'_{u_1}, \dots, T'_{u_i}$ from $T'$ (along with their incident edges) and repeat.

\begin{figure}
    \centering
    \includegraphics[scale=1.2]{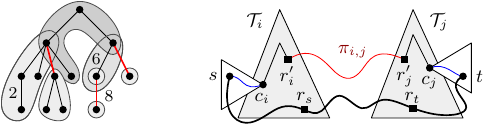}
    \caption{On the left: a tree with edge-weights, where unlabeled edges have weight $1$, and a possible decomposition into micro-trees as computed by our procedure with $W=6$. The edges $(v, u_i)$ are highlighted in red.
    On the right: a qualitative depiction of the paths used in the analysis of the stretch of our group Steiner $(1+\varepsilon)$-spanner. The shortest group Steiner path between $s$ and $t$ is shown in bold, while $\pi_s$ and $\pi_t$ are shown in blue. The white triangles are the trees in $F_i$ and $F_j$ rooted in $c_i$ and $c_j$, respectively.}
    \label{fig:singleton_1_plus_eps}
\end{figure}
 
 We stop the above procedure as soon as $w(T') \le W$, and we choose $T'$ as the last micro-tree of our subdivision (see \Cref{fig:singleton_1_plus_eps}).
  Notice that the edges deleted in each iteration have a total weight of at least $W$, and that each iteration creates at most $2$ micro-trees.
 It follows that the resulting collection contains at most $2 \left( \frac{w(T)}{W} + 1\right) \le \frac{8}{\varepsilon} + 2 = O(\frac{1}{\varepsilon})$ micro-trees.

We then compute a complete graph $\overline{G}$ on the required vertices, where the weight of a generic edge $(u,v)$ is $d_G(u,v)$.
Then, for each unordered pair of (not necessarily distinct) micro-trees $\{\mathcal{T}_i, \mathcal{T}_j\}$, we consider all pairs $(r, r')$ of required vertices such that $r$ is in $\mathcal{T}_i$ and $r'$ is in $\mathcal{T}_j$, compute a minimum-cost Hamiltonian path in $\overline{G}$ between $r$ and $r'$, and we call $\pi_{i,j}$ the shortest of such paths in which each edge of $\overline{G}$ has been replaced by the corresponding shortest path in $G$.
Thanks to \Cref{lm:linear_size_path}, we can assume that each $\pi_{i,j}$ contains at most $O(n)$ edges.

Finally, we compute $h$ forests $F_1, \dots, F_h$, where $F_i$ is obtained by our clustering procedure using the vertices of $\mathcal{T}_i$ as centers. 

Our group Steiner spanner $H$ consists of the union of $T$, all paths $\pi_{i,j}$ for $1 \le i \le j \le h$, and all the edges in $F_i$ for $1 \le i \le h$. The size of $H$ is $O(n/\varepsilon^2)$ as each of (i) $T$, (ii) the $O(\frac{1}{\varepsilon^2})$ paths $\pi_{i,j}$, and (iii) $O(\frac{1}{\varepsilon})$ forests $F_i$, all contain $O(n)$ edges.

To analyze the stretch of $H$, fix a shortest group Steiner path $\pi^*$ between any two vertices $s,t$ in $G$ and let $r_s$ (resp.\ $r_t$) be the first (resp.\ last) required vertex encountered in a traversal of $\pi^*$ from $s$ to $t$.
Let $\mathcal{T}_i$ and $\mathcal{T}_j$ the micro-trees containing $r_s$ and $r_t$, respectively. Moreover, let $r'_i$ and $r'_j$ be the endvertices of $\pi_{i,j}$, where $r'_i$ lies in  $\mathcal{T}_i$ and $r'_j$ lies in $\mathcal{T}_j$.
Finally, let $c_i$ (resp.\ $c_j$) be the root of the tree containing $s$ in $F_i$ (resp.\ $t$ in $F_j$). The situation is depicted in \Cref{fig:singleton_1_plus_eps}. 

Notice that the weight of a minimum-cost group Steiner tree is a lower bound to $w(\pi^*)$, hence $\frac{w(T)}{2} \le \BD_G(s,t)$. Moreover, since $r_s$ is a center of the clustering procedure for $F_i$, we have $d_{F_i}(s,c_i) \le w(\pi^*[s:r_s])$, and symmetrically  $d_{F_i}(c_j,t) \le w(\pi^*[t:r_t])$. Therefore: 
\begin{align}
    \BD_H(s,t) &\le d_{F_i}(s,c_i) + d_{\mathcal{T}_i}(c_i, r'_i) + w(\pi_{i,j}) +  d_{\mathcal{T}_j}(r'_j, c_j) + d_{F_j}(r_j,t) \label{eq:singleton_1_plus_eps_distance_in_H}\\
    & \le w(\pi^*[s:r_s])  +  W  +  w(\pi^*[r'_i: r'_j]) + W   +  w(\pi^*[r_t,t]) \nonumber \\
    & \le  w(\pi^*) + 2 \cdot \frac{\varepsilon}{4} w(T) \nonumber
    \le  \BD_G(s,t) + \varepsilon \BD_G(s,t)=(1+\varepsilon)\BD_G(s,t).
\end{align}

Observe that all steps of the above construction can be carried out in polynomial time, except for the computation of the paths $\pi_{i,j}$, which requires time $2^k k n^{O(1)}$ (see~\Cref{sec:computing_group_Steiner_paths}). Hence, we have the following:

\begin{theorem}
\label{thm:singleton_1_plus_eps}
    In the singleton case, it is possible to compute a group Steiner spanner having stretch $1+\varepsilon$ and size $O(\frac{n}{\varepsilon^2})$ in $2^k k \cdot n^{O(1)}$ time. 
\end{theorem}

To obtain a polynomial building time, we can redefine each $\pi_{i,j}$ starting from a $\gamma$-approximation of minimum-cost Hamiltonian path between (a vertex of) $\mathcal{T}_i$ and (a vertex of) $\mathcal{T}_j$ in $\overline{G}$. We use $w(\pi_{i,j}) \le \gamma w(\pi^*[r'_i,r'_j])$ in~\Cref{eq:singleton_1_plus_eps_distance_in_H} to show that 
\begin{align*}
    \BD_H(s,t) &
    \le w(\pi^*[s:r_s])  +  W  +  \gamma w(\pi^*[r'_i: r'_j]) + W   +  w(\pi^*[r_t,t]) \\
    & \le \gamma w(\pi^*) + 2 \cdot \frac{\varepsilon}{4} w(T)
    \le  \gamma \BD_G(s,t) + \varepsilon \BD_G(s,t)=(\gamma+\varepsilon)\BD_G(s,t).
\end{align*}

We can then state the following:

\begin{theorem}
\label{thm:singleton_gamma_plus_eps}
In the singleton case, it is possible to compute a group Steiner spanner having stretch $\gamma+\varepsilon$ and size $O(\frac{n}{\varepsilon^2})$ in polynomial time, where $\gamma$ is the approximation ratio for the  minimum-cost metric Hamiltonian path problem. 
\end{theorem}

If $G$ is weighted, we can choose $\gamma =\frac{3}{2}$ \cite{Zenklusen19}, while if $G$ is unweighted (i.e., $\overline{G}$ is the metric closure of an unweighted graph) we can choose $\gamma=\frac{7}{5} + \delta$ \cite{SeboV14,TraubVZ20}, for any constant $\delta > 0$.

\subparagraph*{A corresponding distance oracle.} Here we show how to transform the above spanners into group Steiner distance oracles with the same size, stretch and building time. The oracles will be able to answer distance queries in $O(1/\varepsilon^2)$ time and report a corresponding path with an additional time proportional to the number of the path's edges.

We only discuss how to build the oracle with stretch $1+\varepsilon$ since the version with stretch $\gamma+\varepsilon$ is analogous.
Our oracle explicitly maintains all the $O(1/\varepsilon^2)$ paths $\pi_{i,j}$ and their lengths, all the $h=O(1/\varepsilon)$ micro-trees $\mathcal{T}_1,\dots,\mathcal{T}_h$ and the forests $F_1, \dots, F_h$. Moreover, for each micro-tree $\mathcal{T}_i$, we maintain a linear-size data structure  that, given any two vertices $a$ and $b$ of $\mathcal{T}_i$, returns the length of the (unique) path in $\mathcal{T}_i$ between $a$ and $b$ in constant time, and the corresponding path in time proportional to its number of edges.\footnote{For instance, such a data structure can be implemented by rooting $\mathcal{T}_i$ at an arbitrary vertex and using the \emph{least-common-ancestor} data structure in~\cite{BenderF02}.}
Notice that the size of our oracle is $O(n/\varepsilon^2)$. 

To answer a distance query for the vertices $s$ and $t$, we look at all pairs of micro-trees $\mathcal{T}_i,\mathcal{T}_j$. 
For each such pair, arguments analogous to the ones used in the analysis of the stretch factor of the spanner, show that there exists a group Steiner path in $H$ of length
\[
    \ell_{i,j}(s,t) = d_{F_i}(s, c_i) + d_{\mathcal{T}_i}(c_i, r'_i) + w(\pi_{i,j}) +  d_{\mathcal{T}_j}(r'_j, c_j) + d_{F_j}(c_j, t),
\]
where $c_i$ (resp.\ $c_j$) is the root of the tree containing $s$ in $F_i$ (resp.\ $t$ in $F_j$) and $r'_i, r'_j$ are the endpoints of $\pi_{i,j}$. Notice that $\ell_{i,j}(s,t)$ can be evaluated in constant time.
We return $\min_{1 \le i \le j \le h} \ell_{i,j}(s,t)$, which is guaranteed to be at most $(1+\varepsilon) d_{G}(s,t)$ by~\Cref{eq:singleton_1_plus_eps_distance_in_H}.
Once the pair $i,j$ minimizing $\ell_{i,j}(s,t)$ is known, we can also report a group Steiner path of length $\ell_{i,j}(s,t)$ in constant time per path edge by navigating $F_i$ (resp.\ $F_j$) from $s$ (resp.\ $t$) to the root of its tree, and by querying our data structures for the micro-trees $\mathcal{T}_i$ and $\mathcal{T}_j$.

\subsection{A tree spanner with a tight stretch of \texorpdfstring{$2$}{2}}\label{sec:2_spanner_singleton}

\begin{figure}
    \centering
    \includegraphics[width=.6\linewidth]{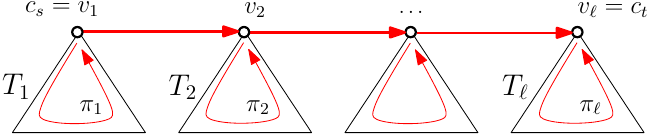}
    \caption{A qualitative depiction of the path constructed in the proof of \Cref{lemma:stretch_2_group_steiner_path_in_T}. The path in bold is the unique path $\pi$ from $c_s$ to $c_t$ in $T$. Each $\pi_i$ is the Eulerian tour of the corresponding tree $T_i$, and the final path is in red.}
    \label{fig:stretch-2-group-steiner-path-in-tree}
\end{figure}

We now describe how to obtain group Steiner \emph{tree spanner} with stretch $2$ (and $n-1$ edges).

We first compute a complete graph $\overline{G}$ on the required vertices, where the weight of a generic edge $(u,v)$ is $d_G(u,v)$. We then compute an MST $M$ of $\overline{G}$, and we construct a subgraph $\widetilde{M}$ of $G$ by replacing each edge $(u,v)$ in $M$ with a shortest path between $u$ and $v$ in $G$. Finally, we select any spanning tree $T$ of $\widetilde{M}$. Our spanner $H$ is obtained as the union of $T$ with the forest $F$ computed by our clustering procedure using the vertices in $T$ as centers.

Notice that $H$ has $n-1$ edges since it is a spanning tree of $G$. We now show that the stretch factor of $H$ is at most $2$. Consider any pair of vertices $s,t$ and let $c_s$ and $c_t$ be the roots of the trees of $F$ containing $s$ and $t$, respectively (notice that $c_s$ and $c_t$ might coincide).

\begin{lemma}\label{lemma:stretch_2_group_steiner_path_in_T}
There exists a path between $c_s$ and $c_t$ in $T$ that traverses all vertices of $T$ and has length $2w(T) - d_T(c_s, c_t)$.
\end{lemma}
\begin{proof}
    Let $\pi = \langle c_s = v_1, v_2, \dots, v_\ell = c_t \rangle$ be the unique path from $c_s$ to $c_t$ in $T$ (if $c_s = c_t$ then $\pi = \langle c_s \rangle$),
    and call $T_i$ the unique tree containing $v_i$ in the forest obtained from $T$ by deleting the edges of $\pi$.
    Let $\pi_i$ be an Eulerian tour of $T_i$ that starts and ends in $v_i$ and observe that $\pi_i$ traverses each edge in $T_i$ twice.
    The sought path is obtained by joining all tours $\pi_1, \dots, \pi_\ell$ with the edges of $\pi$, i.e., it is the path $\pi_1 \circ (v_1, v_2) \circ \pi_2 \circ (v_2, v_3) \circ \dots \circ (v_{\ell-1}, v_\ell) \circ \pi_\ell$ (see \Cref{fig:stretch-2-group-steiner-path-in-tree}).
\end{proof}

Using the above lemma, we consider the path from $s$ to $t$ in $H$ consisting of the composition of (i) the unique path $\pi_s$ from $s$ to $c_s$ in $F$, (ii) a path $\pi$ from $c_s$ to $c_t$ of length at most $2w(T)$ that traverses all vertices in $T$ (whose existence is guaranteed by \Cref{lemma:stretch_2_group_steiner_path_in_T}), and (iii) the unique path $\pi_t$ from $c_t$ to $t$.

We now argue that the length of the above path is at most $2 \BD(s,t)$. Let $\pi^*$ be the optimal group Steiner path from $s$ to $t$ and let $r_1^*$ and $r_k^*$ be the first and the last occurrences of a required vertex in $\pi^*$, respectively. Since $r_1^*$ is a vertex of $T$, the length of $\pi^*[s:r^*_1]$ is at least $w(\pi_s)$. Similarly, the length of $\pi^*[r^*_k: t]$ is at least $w(\pi_t)$.
Moreover, observe that the weight of $M$ is a lower bound for the length of $\pi^*[r^*_1: r^*_k]$, and hence $w(T) \le  w(M) \le w(\pi^*[r^*_1: r^*_k])$. Therefore:
\begin{align}
 \BD_H(s,t) 
 &\le w(\pi_s) + w(\pi) + w(\pi_t)
 = d_{F}(s, c_s) + 2w(T) - d_T(c_s, c_t) + d_{F}(c_t, t) \label{eq:stretch_2_distance}\\
 & \le w(\pi^*[s:r^*_1]) + 2w(T)  + w(\pi^*[r^*_k: t]) \nonumber \\
 & \le w(\pi^*[s:r^*_1]) + w(\pi^*[r^*_1:r^*_k]) + w(\pi^*[r^*_k: t]) + w(T) \nonumber \\
 &\le w(\pi^*) + w(T) 
 \le 2 w(\pi^*) = 2 \BD_G(s,t). \nonumber
\end{align}

\begin{theorem}
    \label{thm:singleton_stretch_2}
    In the singleton case, it is possible to compute a group Steiner tree spanner having stretch $2$ and $n-1$ edges in polynomial time.
\end{theorem}

We can prove that the stretch of the above tree spanner is tight, since its stretch cannot be improved even for the single-source case. 

\begin{theorem}
\label{thm:singleton_lb_single_source}
    In the singleton case, there are unweighted graphs $G$ such that any single-source group Steiner spanner of $G$ having stretch strictly smaller than $2-\frac{2}{k}$ must contain at least $n$ edges.
\end{theorem}
\begin{proof}
    To prove our lower bound consider a graph $G$ consisting of a cycle $C$ on the $k$ required vertices $r_1, \dots, r_k$, plus $n-k$ additional vertices $v_1, \dots, v_{n-k}$ connected to an arbitrary required vertex $r$ via the edges in $F = \{(v_1, r), \dots, (v_{n-k}, r)\}$, and let $r_1$ be the source vertex (see \Cref{fig:lb_singleton}). 
    Clearly all single-source group Steiner spanners of $G$ need to contain all edges in $F$, since otherwise they would be disconnected.
    Consider now any subgraph $H$ obtained from $G$ by deleting a generic edge $e$ of the cycle. 
    Observe that the shortest group Steiner path from $s$ to itself in $H$ requires traversing each edge $C - e$ twice, hence $\BD_H(s,s) = 2(k-1)=2k-2$.
    Since the shortest group Steiner path in $G$ from $s$ to itself has length $k$,  the stretch factor of $H$ is at least $\frac{2k-2}{k} = 2 - \frac{2}{k}$.
\end{proof}

\begin{figure}
    \centering
    \includegraphics[scale=0.7]{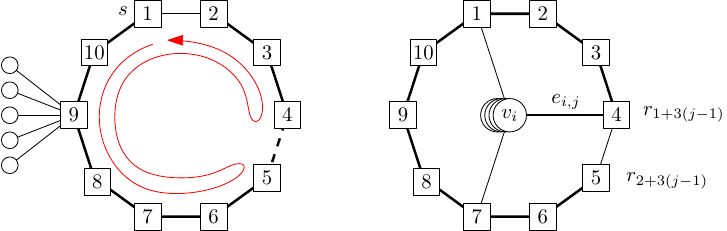}
    \caption{The lower bound constructions of \Cref{thm:singleton_lb_single_source} (left) and \Cref{thm:singleton_preserver} (right) with $k=10$ required vertices. The generic required vertex $r_i$ is depicted as a squared labelled with $i$.}
    \label{fig:lb_singleton}
\end{figure}

\subsubsection*{A corresponding distance oracle}

The above group Steiner $2$-spanner can be turned into a distance oracle  with constant query time by storing $F$ and, for vertex $v \in V$, the distance $d_F(v,c_v)$ from $v$ to the root $c_v$ of its tree in $F$. We also store $T$ and $w(T)$, along with a data structure that can report distances in $T$ in constant time, and the associated paths in an additional time proportional to the number of returned edges. 

To answer a distance query we compute $d_{F}(s, c_s) + 2w(T) - d_T(c_s, c_t) + d_{F}(c_t, t)$ (see~\Cref{eq:stretch_2_distance}) in constant time.
To report the associated path we return $\pi_s \circ \pi \circ \pi_t$ as defined above, where $\pi_s$ (resp.\ $\pi_t$) can be found by navigating $F$ from $s$ (reps.\ $t$) towards the root of its tree, and $\pi$ can be found by a suitable DFS visit of $T$ from $c_s$ that prioritizes edges that are not in the unique path from $c_s$ to $c_t$ in $T$ (see the proof of \Cref{lemma:stretch_2_group_steiner_path_in_T}). This can be done time proportional in the number of edges of $\pi_s \circ \pi \circ \pi_t$.

\section{Group Steiner spanners for general group sizes}
\label{sec:group_steiner_spanner_general}

We now describe how to obtain an all-pairs group Steiner spanner $H$ with stretch factor $2\alpha + 1$, for each $\alpha \geq 1$. We will provide two different constructions that build different spanners of sizes $kn+|\bigcup_i R_i \times R_i\ \alpha\text{-spanner}|$ and $n+|R\times R\ \alpha\text{-spanner}|$, respectively. The first construction is preferable over the second one when the $k$ groups are somewhat disjoint and of uniform sizes, i.e., $|R_i|=O(|R|/k)$ for all $i$. We also provide corresponding distance oracles having query a time of $O(2^k k \cdot |R|^2+|R|^3)$.

\subsection{The first construction.} The first construction is the following.
For each group $R_i \subseteq V$, we build the spanner $H$ as the union of a \emph{subsetwise} spanner $R_i \times R_i$ with stretch $\alpha$, plus $k$ spanning forests $F_1, \dots, F_k$, where each $F_i$ is obtained by our clustering procedure using the vertices in $R_i$ as centers.

We now discuss the stretch factor of $H$.
Fix any two vertices $s,t \in V$, and, without loss of generality, assume that the shortest group Steiner path $\pi^*$ from $s$ to $t$ in $G$ traverses the groups $R_1,\dots,R_k$ in this order (otherwise, re-index the groups accordingly), and let $r_i$ be the first required vertex in $R_i$ reached by $\pi^*$ (see \Cref{fig:2alpha+1_spanner}).

Let $r'_1$ (resp. $r'_{k+1}$) be the root of the tree in $F_1$ (resp. $F_k$) containing $s$ (resp. $t$) and let $\pi_1$ (resp. $\pi_{k+1}$) the unique path in $F_1$ (resp. $F_k$) between $s$ and $r'_1$ (resp. between $t$ and $r'_{k+1}$). Moreover, for each $i=2,\dots,k$, let $r'_i$ be the root of the tree in $F_i$ containing $r_{i-1}$, and let $\pi_i$ be the corresponding path in $F_i$ between them.
Finally, denote by $\pi'_i$ the path in $H$ between $r'_i$ and $r_i$, for $i=1,\dots,k$, and by $\pi'_{k+1}$ the path in $H$ between $r_k$ and $r'_{k+1}$.
Consider now the group Steiner path $\widetilde{\pi}$ in $H$ made by the concatenation of $\pi_1 \circ \pi'_1 \circ \dots \circ \pi_k \circ \pi'_k \circ \pi'_{k+1} \circ \pi_{k+1}$ (see \Cref{fig:2alpha+1_spanner}).
We now show that $w(\widetilde{\pi}) \le (2\alpha+1) w(\pi^*)$. 

For technical convenience, we let $r_0 = s$ and $r_{k+1} = t$ and we notice that, for $i = 1, \dots, k+1$, we have $w(\pi_i) \le w(\pi^*[r_{i-1} : r_i])$. 

For $i=1, \dots, k$, $w(\pi'_i) \le \alpha d_G(r'_i, r_i) \le \alpha ( w(\pi_i) +  w(\pi^*[r_{i-1} : r_i]) ) \le 2 \alpha  w(\pi^*[r_{i-1} : r_i]) $.
Finally, $w(\pi'_{k+1}) \le \alpha d_G(r_k, r'_{k+1}) \le \alpha (w(\pi^*[r_k : r_{k+1}]) + w(\pi_{k+1}) ) \le 2 \alpha w(\pi^*[r_k : r_{k+1}])$.

\begin{align*}
    \BD_H(s,t)
    & \le w(\widetilde{\pi}) = \sum_{i=1}^{k+1} \left(  w(\pi_i) + w(\pi'_i) \right) 
    \le \sum_{i=1}^{k+1} w(\pi^*[r_{i-1} : r_i]) + 2 \alpha \sum_{i=1}^{k+1} w(\pi^*[r_{i-1} : r_i]) \\
    & = w(\pi^*) + 2\alpha w(\pi^*)
    = (1 + 2\alpha) \BD_G(s,t).
\end{align*}

\begin{theorem}
\label{thm:all_pairs_general_2alpha_plus_1_spanner}
Given $k$ subsetwise spanners $H_1, \dots, H_k$, where $H_i$ is an $R_i \times R_i$ $\alpha$-spanner of $G$, it is possible to compute in polynomial time a group Steiner spanner of $G$ with stretch $2\alpha+1$ and size $O\left(nk + \left| \bigcup_{i=1}^k H_i \right| \right)$. 
\end{theorem}

\begin{figure}
    \centering
    \includegraphics[scale=.60]{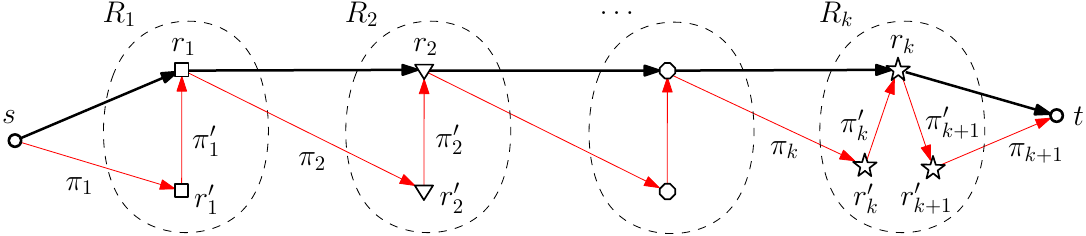}
    \caption{A qualitative depiction of analysis of the stretch $2 \alpha+1$. The shortest group Steiner path $\pi^*$ from $s$ to $t$ is in bold, while the path $\widetilde{\pi}$ in the spanner is in red.}
    \label{fig:2alpha+1_spanner}
\end{figure}

\subsubsection*{A corresponding distance oracle}

For each vertex $v$, we store the root $r_{i,v}$ of the tree containing $v$ in the forest $F_i$ computed by our clustering procedure using the vertices in $R_i$ as centers, along with $d_G(v, r_{i,v}) = d_{F_i}(v, r_{i,v})$.
Then, we compute $k$ complete auxiliary graphs $G_{R_i}$ for $i=1,\dots,k$ on the vertices in $R_i$, where the weight of a generic edge $(u,v)$ is $d_G(u,v)$.
Next, we build a graph $H_R$ which is obtained as the union, for each $i=1,\dots,k$, of a
$R_i \times R_i$ $\alpha$-spanner of $G_{R_i}$ (w.r.t.\ shortest path metric). 
This graph $H_R$ is then augmented by adding, for each required vertex $x \in R$, the $k$ edges $(x, r_{i,x})$ for $i=1, \dots, k$, where the weight of $(x, r_{i,x})$ is $d_G(x, r_{i,x})$.

Notice that $H_R$ (and its augmented version) can be computed in polynomial time and that its worst-case asymptotic size is at most that size of our group Steiner spanner.

To report $\BD(s,t)$ we construct a new graph $H'_R$ by adding new vertices and edges to the augmented version of $H_R$. In details, we add two new vertices $s'$ and $t'$, along with all edges $(s', r_{i,s})$ and $(t', r_{i, t})$ of respective weights $d_G(s, r_{i,s})$ and $d_G(t, r_{i,t})$ for $i=1,\dots,k$. Finally, we return the group Steiner distance $\BD_{H'_R}(s', t')$ as computed by the algorithm of \Cref{sec:computing_group_Steiner_paths}.
This requires time $O(2^k k \cdot |R|^2 + |R|^3)$.

\subsection{The second construction.}
The second construction is the following. We build the spanner $H$ as the union of a subsetwise $R\times R$ $\alpha$-spanner $H'$ of $G$, plus a spanning forest $F$ that is obtained by our clustering procedure using the vertices in $R$ as centers. Clearly, the size of $H$ is $O(n+|H'|)$.

We now show that the stretch factor of $H$ is $2\alpha + 1$.
Fix any two vertices $s,t \in V$, and, without loss of generality, assume that the shortest group Steiner path $\pi^*$ from $s$ to $t$ in $G$ traverses the groups $R_1,\dots,R_k$ in this order (otherwise, re-index the groups accordingly). Let $r_i$ be the first required vertex in $R_i$ reached by $\pi^*$. Let $r_s$ be the vertex of $R$ that corresponds to the center of the cluster in $F$ that contains $s$. Similarly, let $r_t$ be the vertex of $R$ that corresponds to the center of the cluster in $F$ that contains $t$ (it might happen that $r_s=r_t$). We have $d_H(s,r_s)=d_G(s,r_s) \leq d_G(s,r_1)$ and $d_H(r_t,t)=d_G(r_t,t)\leq d_G(r_k,t)$. Moreover, we can upper bound the group Steiner distance from $r_s$ to $r_t$ in $G$ by the following 
\[
\BD_G(r_s,r_t)\leq d_G(r_s,s)+\BD_G(s,t)+d_G(t,r_t) \leq d_G(s,r_1)+\BD_G(s,t)+d_G(r_k,t)\leq 2\BD_G(s,t).
\]
As a consequence, since $H$ contains an $R\times R$ $\alpha$-spanner of $G$, using~\Cref{lemma:group_steiner_path_decomposition},  we obtain that $\BD_H(r_s,r_t) \leq \alpha \BD_G(r_s,r_t)\leq 2\alpha \BD_G(s,t)$. Therefore $\BD_{H}(s,t)$ is at most:
\[
   \leq d_H(s,r_s)+\BD_H(r_s,r_t)+d_H(r_t,t)                \leq d_G(s,r_1)+2\alpha \BD_G(s,t) + d_G(r_k,t) \leq (2\alpha+1) \BD_G(s,t).
\]

\begin{theorem}
\label{thm:all_pairs_general_2alpha_plus_1_spanner_bis}
Given a subsetwise $R \times R$ $\alpha$-spanner $H'$ of $G$, it is possible to compute in polynomial time a group Steiner spanner of $G$ with stretch $2\alpha+1$ and size $O\left(n + |H'| \right)$. 
\end{theorem}

\subsubsection*{A corresponding distance oracle}

For each vertex $v$, we store the root $r_v$ of the tree containing $v$ in the forest $F$ computed by our clustering procedure, along with $d_G(v, r_v) = d_F(v, r_v)$.
Then, we compute a complete auxiliary graph $G_R$ on the vertices in $R$, where the weight of a generic edge $(u,v)$ is $d_G(u,v)$.
Next, we compute and we store a $R \times R$ $\alpha$-spanner $H_R$ of $G_R$ (w.r.t.\ shortest path metric). Notice that both $G_R$ and $H_R$ can be computed in polynomial time and that the asymptotic size of $H_R$ is at most the size of our group Steiner spanner in the worst case.

To report $\BD(s,t)$ we compute $\BD_{H_R}(r_s, r_t)$ and we return $d_G(s, r_s) + \BD_{H_R}(r_s, r_t) + d_G(r_t, t)$. This can be done in time $O(2^k k \cdot |R|^2 + |R|^3)$ using the algorithm of \Cref{sec:computing_group_Steiner_paths}.

\section{Single-source group Steiner spanners}
\label{app:single-source}

In this section we present single-source group Steiner spanners, some of which can be converted into distance oracles with constant query time. 

\subsection{A single-source group Steiner preserver with size \texorpdfstring{$O(2^k n)$}{O(2\^k n)}}

We first focus on the design of single-source preservers. We observe that a single-source group Steiner preserver only needs to preserve $n$ shortest group Steiner paths, one for each target vertex $t \in V$. As \Cref{lemma:group_steiner_path_decomposition} states that each such path can be obtained by concatenating $k+1$ shortest paths in $G$, we can build single-source group Steiner preservers by constructing pairwise preservers (w.r.t. shortest-path metric) for a set $P$ of $O(kn)$ pairs of vertices. Thus, for weighted graphs, the preserver of size $O(\min\{|P|n^{1/2},n|P|^{1/2}\})$ given in~\cite{CoppersmithE06} implies the existence of a single-source group Steiner preserver of size $O(k^{1/2}n^{3/2})$. For unweighted graphs, the preserver of size $O(n^{2/3}|P|^{2/3}+n|P|^{1/3})$ given in~\cite{BodwinWilliams} implies the existence of a single-source group Steiner preserver of size $O(k^{2/3}n^{4/3})$. It is worth noticing that, in general, we do not know how to compute the set $P$ in polynomial time. Therefore, to the best of our knowledge, the construction of single-source pairwise preservers from classical preservers requires $2^k k \cdot n^{O(1)}$ time as, to compute the set $P$, we need to run the algorithm that computes a shortest group Steiner path from $s$ to every target vertex $t$. 

In this section, we provide an algorithm that builds a single-source group Steiner \emph{preserver} with at most $2^k n$ edges in time $2^k k \cdot n^{O(1)}$. Notably, our construction implies the existence of single-source group Steiner preserver of size $O(n \, \text{polylog}(n))$ when $k=O(\log \log n)$.

Recall that $\BD(s,t \mid \mathcal{C})$ denotes the group Steiner distance between $s$ and $t$ w.r.t.\ the groups in $\mathcal{C}$.
We compute $H$ incrementally, starting from $H=(V,\emptyset)$. We will inductively process each subset $\mathcal{C} \subseteq \mathcal{R}$ in increasing order of size $\vert \mathcal{C} \vert$, and, after processing $\mathcal{C}$, we will guarantee the following properties: 
\begin{enumerate}[i.]
    \item we have explicitly computed and stored all the values $\BD_G(s,t \mid \mathcal{C})$, for every $t$;
    
    \item the current spanner $H$ is a single-source preserver for the groups $\mathcal{C}$, i.e. $\BD_H(s,t \mid \mathcal{C}) = \BD_G(s,t \mid \mathcal{C})$.
\end{enumerate}

Obviously, when $\mathcal{C} = \mathcal{R}$ the subgraph $H$ is guaranteed to be a single-source group Steiner preserver with source vertex $s$ w.r.t.\ the groups $R_1,\dots, R_k$.

We now show how to process the generic group set $\mathcal{C} \subseteq \mathcal{R}$.
The base case is $\mathcal{C} = \emptyset$.
In this case we add to $H$ the shortest path tree rooted in $s$  and we store all distances from $s$ to the vertices of $V$ in $G$. Clearly, $H$ satisfies properties (i) and (ii).

Consider now a non-empty collection $\mathcal{C}=\{ \widetilde{R}_1, \dots, \widetilde{R}_h \}$.
Notice that, by induction,  $H$ is actually a single-source preserver w.r.t.\ every $\mathcal{C'} \subset \mathcal{C}$. We show how to add at most $n-1$ edges to $H$, in order to guarantee (i) and (ii) for $\mathcal{C}$ as well.

We create a new graph $\widetilde{G}$ by augmenting $G$ with a new vertex $\widetilde{s}$ and an edge from $\widetilde{s}$ for each required node $r \in \bigcup_{i=1}^h \widetilde{R}_i$, of weight $\BD_G(s,r \mid \mathcal{C} \setminus \{\widetilde{R}_i \mid r \in \widetilde{R}_i \})$ (see \Cref{fig:single-source-preserver}). Notice that all the weights of the added edges are values that we have already computed (by property (i)).

We then find a shortest path tree $T$ of $\widetilde{G}$ with source $\widetilde{s}$, and we store the corresponding $n$ single-source distance values $d_T(\widetilde{s}, v)$ for $v \in V$.
\Cref{lemma:group_steiner_path_decomposition} guarantees that, for every vertex $t \in V$, we have:
\begin{align*}
    d_T(\widetilde{s},t)=d_{\widetilde{G}}(\widetilde{s},t)
    &= \min_{r \in \cup_{i=1}^h \widetilde{R}_i} \left( w(\widetilde{s},r) + d_G(r,t) \right) \\
    &= \min_{1 \leq i \leq h} \min_{r \in \widetilde{R}_i} \left( \sigma_G(s,r \mid \mathcal{C} \setminus \{\widetilde{R}_j \mid r \in \tilde{R}_j\}) + d_G(r,t) \right)
    = \sigma_G(s,t \mid \mathcal{C}).
\end{align*}

Moreover, property (ii) and the induction hypothesis ensure that, for each $r \in \widetilde{R}_i$,  $H$ contains a group Steiner path of length $\sigma_G(s,r \mid \mathcal{C} \setminus \{\widetilde{R}_j \mid r \in \tilde{R}_j\})$ w.r.t.\ the groups in $\mathcal{C} \setminus \{\widetilde{R}_j \mid r \in \tilde{R}_j\}$. Hence, by adding to $H$ the spanning forest $F_{\mathcal{C}}$ obtained by removing $\widetilde{s}$ (and its incident edges) from $T$, we have that $H$ satisfies properties (i) and (ii) w.r.t.\ the groups in $\mathcal{C}$. 

This proves that at the end of the algorithm, $H$ is a single-source group Steiner preserver of $G$. 
Moreover, the size of $H$ is $O(2^k n)$ as it increases by at most $n-1$ for each of the $n-1$ iterations. Since each iteration can be implemented in polynomial time, we have proved the following: 

\begin{figure}
    \centering
    \includegraphics[scale=.7]{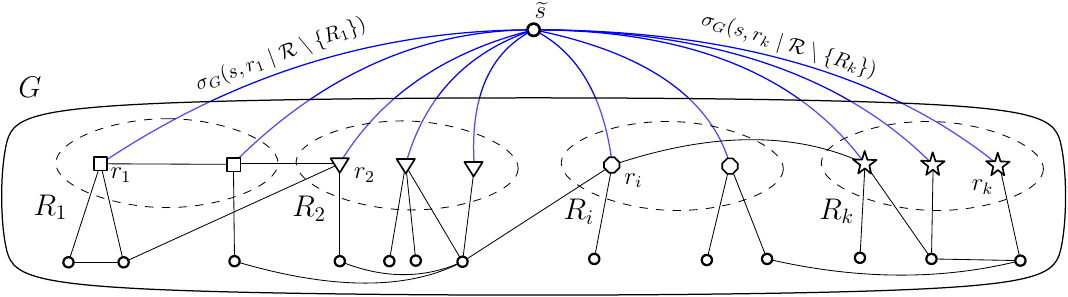}
    \caption{The auxiliary graph $\widetilde{G}$ for the case $\mathcal{C} = \mathcal{R}$.}
    \label{fig:single-source-preserver}
\end{figure}

\begin{theorem}
\label{thm:single_source_preserver}
It is possible to compute a single-source group Steiner preserver having at most $2^kn$ edges in time $2^k k \cdot n^{O(1)}$. 
\end{theorem}

\subparagraph*{A corresponding distance oracle.}

An oracle reporting exact group Steiner distances in constant time and having size $O(n)$ is trivial to obtain in the single-source case, once the above spanner has been computed since it suffices to store the values $\BD(s,t \mid \mathcal{R})$.
The same asymptotic size of the spanner allows for a more general query: we can report the group Steiner distances from $s$ w.r.t.\ any subset of groups of choice by simply storing $\BD(s,t \mid \mathcal{C})$ for every $\mathcal{C} \subseteq \mathcal{R}$.

We can also provide the path associated with this latter query by storing a shortest-path tree of $G$ from $s$ and all forests $F_\mathcal{C}$ computed while building the spanner.
To answer a path query for $t$ w.r.t.\ the groups in $\mathcal{C} \subseteq \mathcal{R}$, we find the root $r$ of the tree containing $t$ in $F_\mathbf{C}$, and the unique path $\pi_t$ from $r$ to $t$ in $F_\mathcal{C}$. We then recursively query our oracle for the path $\pi_r$ between $s$ and $t$ w.r.t.\ the groups in $\mathcal{C} \setminus \{ R_i \mid r \in R_i \}$, and report $\pi_r \circ \pi_t$. The base case of the recursion correspond to $\mathcal{C} = \emptyset$, in which we return the shortest path $\pi_t$ from $s$ to $t$ in $G$.

Each recursive call can be implemented to run in constant time, plus a time proportional to the number of edges of the last portion $\pi_t$ of the returned path. One can observe that the number of returned edges is asymptotically lower bounded by the number of recursive calls, hence to overall time spent is proportional to the number of edges of the returned path.

\subsection{A group Steiner tree spanner with stretch \texorpdfstring{$3$}{3}}

Let $\pi^*$ be the shortest group Steiner path starting from $s$ among those that end in a required vertex, say $r^*$. More formally, $r^*= \arg \min_{r \in R} \BD_G(s,r)$. Notice that $\pi^*$ can be computed in time $2^k k \cdot \text{poly}(n)$.

Let $T'$ be any spanning subgraph of $\pi^*$. We define our single-source tree spanner $T$ as the union of $T'$ and the forest $F$ obtained by running the clustering procedure with centers corresponding to all the vertices of $T'$. Clearly, $T$ has $n-1$ edges. We now prove the claimed stretch factor of $3$. For any vertex $t$, let $t'$ be the center in $T'$ of the cluster in $F$ that contains $t$. 
$T$ contains the group Steiner path $\pi$ from $s$ to $t$ which is obtained as follows: starting from $s$, we traverse all edges of $T'$ at most twice to obtain a group Steiner path from $s$ to $t'$ of total weight of at most $2w(\pi^*)$, then we go along the shortest path from $t'$ to $t$ whose length is at most $d_G(s,t)\leq \BD_G(s,t)$ as $s$, being a vertex of $T'$, is a center in the clustering procedure and therefore, $d_G(t',t)\leq d_G(s,t)$. As $w(\pi^*) = \BD_G(s,r^*) \leq \BD_G(s,t)$, we have
\begin{align*}
    w(\pi)
    &\leq 2w(T') + d_F(t',t) \le 2\BD_G(s,r^*) + d_G(t',t) \leq 2\BD_G(s,t) + d_G(s,t) \le 3\BD_G(s,t). 
\end{align*}

\begin{theorem}
    \label{thm:single_source_stretch_3}
    It is possible to compute a single-source group Steiner tree spanner with stretch factor $3$ and $n-1$ edges in 
    time $2^k k \cdot n^{O(1)}$.
\end{theorem}

\subparagraph*{A corresponding distance oracle.} By storing $w(T')$, all the single-source distances $d_{T'}(s,t')$ in $T'$, and the intra-cluster distances $d_F(t',t)$ from any vertex $t$ and its associated center $t'$ in $F$, it is immediate to answer a distance query in constant time by reporting $2w(\pi^*) - d_{T'}(s,t')+ d_F(t',t)$. Moreover, the corresponding group Steiner path can be reported in time proportional to its number of edges as we need to traverse all edges of $T'$ but those along the path from $s$ to $t'$ exactly twice, while all edges along the path from $s$ to $t$ in $T$ needs to be traversed exactly once.

\subsection{A group Steiner spanner with stretch \texorpdfstring{$\alpha + 1$}{alpha+1}}
We build a single-source group Steiner spanner $H$ of stretch $\alpha + 1$ as follows: $H$ is the union of a shortest path tree $T$ of $G$ rooted in $s$, a subsetwise $R\times R$ $\alpha$-spanner $H'$ of $G$, and the forest $F$ obtained by running the clustering procedure with centers in $R$.
We now argue that the stretch factor of $H$ is $\alpha+1$. 

Consider a target vertex $t$, and consider a shortest group Steiner path $\pi^*$ between $s$ and $t$. 
Clearly, $\pi^*$ traverses all the groups in $\mathcal{R}$.
Without loss of generality, assume that $\pi^*$ traverses the groups $R_1,\dots,R_k$ in this order (otherwise, re-index the groups properly), and let $r_i$ be the first required vertex in $R_i$ reached by $\pi^*$. The first upper bound on $\BD_H(s,t)$ is 
\begin{align*}
    \BD_H(s,t)
    &\leq d_T(s,r_1) + d_{H'}(r_1,r_2) + \dots + d_{H'}(r_{k-1},r_k)+d_T(r_k,s)+d_T(s,t)\\
    &\le \alpha d_G(s,r_1)+ \alpha d_G(r_1,r_2)+ \cdots + \alpha d_G(r_{k-1},r_k)+ d_G(r_k,s)+d_G(s,t) \\
    &\le \alpha\big(\BD_G(s,t)-d_G(r_k,t)\big)+\big(\BD_G(s,t)-d_G(r_k,t)\big)+\BD_G(s,t)\\
    &=(\alpha +2)\BD_G(s,t)-(\alpha+1)d_G(r_k,t).
\end{align*}

Let $r_t$ be the vertex of $R$ which is the center of the cluster in $F$ that contains $t$. By the triangle inequality, $d_{H'}(r_k,r_t)\leq \alpha d_G(r_k,r_t)\leq \alpha \big(d_G(r_k,t)+d_G(r_t,t)\big)\leq 2\alpha d_G(r_k,t)$. The second upper bound on $\BD_H(s,t)$ is
\begin{align*}
    \BD_H(s,t)
    &\leq d_T(s,r_1) + d_{H'}(r_1,r_2) + \dots + d_{H'}(r_{k-1},r_k)+d_{H'}(r_k,r_t)+d_F(r_t,t)\\
    &\le \alpha d_G(s,r_1)+ \alpha d_G(r_1,r_2)+ \cdots + \alpha d_G(r_{k-1},r_k)+2\alpha d_G(r_k,r_t)+d_G(r_t,t) \\
    &\le \alpha \big(\BD_G(s,t)-d_G(r_k,t)\big)+2\alpha d_G(r_k,t)+d_G(r_k,t)\\
    &=\alpha\BD_G(s,t)+(\alpha+1)d_G(r_k,t).
\end{align*}
By combining the two upper bounds, we obtain
\[
2\BD_H(s,t) \leq (\alpha+2)\BD_G(s,t)-(\alpha+1)d_G(r_k,t)+\alpha \BD_G(s,t)+(\alpha+1)d_G(r_k,t) = 2(\alpha+1)\BD_G(s,t), 
\]
from which we derive $\BD_H(s,t)\leq (\alpha+1) \BD_G(s,t)$. We have proved the following:

\begin{theorem}
\label{thm:single_source_general_alpha_plus_2_spanner}
Given a subsetwise $R \times R$ $\alpha$-spanner $H'$ of $G$, it is possible to compute in polynomial time a single-source group Steiner spanner of $G$ having stretch $\alpha+1$ and size $O\left(n + \left| H' \right| \right)$. 
\end{theorem}

\subparagraph*{A corresponding distance oracle.}

For each vertex $v$, we store the root $r_v$ of the tree containing $v$ in the forest $F$ computed by our clustering procedure, along with $d_G(v, r_v) = d_F(v, r_v)$.
Then, we compute a complete auxiliary graph $G_R$ on the vertices in $R$, where the weight of a generic edge $(u,v)$ is $d_G(u,v)$.
Next, we compute and we store a $R \times R$ $\alpha$-spanner $H_R$ of $G_R$ (w.r.t.\ shortest path metric). Notice that both $G_R$ and $H_R$ can be computed in polynomial time and that the asymptotic size of $H_R$ is at most the size of our group Steiner spanner in the worst case.

To report $\BD(s,t)$, we construct a new graph $H'_R$ by augmenting $H_R$ via the addition of a new vertex $s'$ along with an edge $(s', r)$ of weight $d_G(s, r)$ for all $r \in R$.
Then we compute all group Steiner distances $\BD_{H'_R}(s', r)$ for $r \in R$ in time $O(2^k k \cdot |R|^2 + |R|^3)$ using the algorithm of \Cref{sec:computing_group_Steiner_paths}.
Finally, we return:
\[
    \min\left\{ \min_{r \in R} \left( \BD_{H'_R}(s', r) + d_G(s', r) + d_G(s,t) \right), \BD_{H'_R}(s, r_t) + d_G(r_t, t) \right\}.
\]

\section{Conclusions}

We conclude this work by mentioning some problems that we deem significant. 
Our construction of the $(1+\varepsilon)$-spanner with size $O(n/\varepsilon^2)$ for the singleton case requires a building time of $2^k k\cdot n^{O(1)}$, can a spanner with the same stretch and size $O(f(\varepsilon) \cdot n \, \text{polylog}(n))$ be built in polynomial time?
Regarding the single-source case, we conjecture that the dependency on $k$ in the $O(2^k \cdot n)$ size of our preserver is too weak. Can stronger upper bounds be proved for either the same or novel constructions? Can a lower bound that is polynomial in $k$ and linear in $n$ be shown?

\bibliographystyle{plainurl}
\bibliography{bibliography.bib}

\clearpage

\appendix

\FloatBarrier

\section{The complexity of finding short group Steiner paths}
\label{app:complexity_gsp}

In this section we discuss the complexity of finding short group Steiner paths and provide both negative results and a FPT algorithm w.r.t.\ the number $k$ of groups. 

\subsection{Hardness and Inapproximability}

\subsubsection{General group sizes} Unsurprisingly, the problem of finding shortest group Steiner paths inherits strong inapproximability results from the closely related group Steiner tree problem, where the goal is that of computing a minimum-weight tree in $G$ that contains at least one vertex for each group. 

Indeed, one can observe that any group Steiner path $\pi$ w.r.t.\ the groups in $\mathcal{R}$ implies the existence of a group Steiner tree w.r.t.\ $\mathcal{R}$ of weight at most $w(\pi)$ which can be found by computing any spanning tree of the subgraph of $G$ induced by the edges in $\pi$.
Conversely, a group Steiner tree $T$ of $G$ w.r.t.\ $\mathcal{R} \cup \{ \{s\}, \{t\} \}$ implies the existence of a group Steiner path between $s$ and $t$ of $G$ w.r.t.\ $\mathcal{R}$ of length at most $2w(T)$ which is a portion of a suitable Eulerian tour of $T$. More precisely, all edges of $T$ except those in the unique path $P$ from $s$ to $t$ are traversed twice, while those in $P$ are traversed only once.
As a consequence, the group Steiner path and the group Steiner tree problems have the same asymptotic (in-)approximability.

Halperin and Krauthgamer~\cite{halperinSTOC03} showed that, given any constant $\varepsilon>0$, there exists no polynomial-time $O(\log^{2-\varepsilon} k)$-approximation algorithm to compute a minimum-weight group Steiner tree, unless $\NP \subseteq \mathsf{ZPTIME}(n^{\text{polylog}(n)})$, and the above holds even if the input graphs are trees. As discussed above, this translates into the same inapproximability result for the problem of computing a shortest group Steiner path. Additionally, the authors of~\cite{halperinSTOC03} prove their result by showing a gap-producing reduction from SAT, which implies that the above hardness carries over to the problem of computing group Steiner distances.

\subsubsection{The singleton case}\label{apx:steiner_path_singleton_case}

We now discuss the relation between the problem of computing a shortest group Steiner path in the singleton case and the minimum-cost metric Hamiltonian path problem. 
In this problem we are given a complete edge-weighted graph on $\eta$ vertices which satisfies the triangle inequality, along with two distinguished vertices $s,t$, and the goal is that of finding a Hamiltonian path having $s$ and $t$ as endvertices and minimum total cost (measured as the sum of the weights of the edges in the path). 

The minimum-cost metric  Hamiltonian path problem is known to be not approximable in polynomial time within a factor of $\frac{185}{184}$, unless $\p=\NP$ \cite{KarpinskiLS15}, while it admits a polynomial-time $\frac{3}{2}$-approximation algorithm \cite{Zenklusen19}.

It turns out that the same (in)approximability results also hold for the problem of computing shortest group Steiner paths in the singleton case. 
Indeed, any shortest group Steiner path $\pi$ between $s$ and $t$ is associated with a minimum-cost Hamiltonian path $\widetilde{\pi}$ in the metric complete graph $\widetilde{G}$ having $R \cup \{s,t\}$ as its vertex set and in which each edge $(u,v)$ has weight $d_G(u,v)$.
The edges in $\widetilde{G}$ correspond to shortest paths between their endpoints in $G$, and vice-versa. Conversely, a minimum-cost Hamiltonian path $\widetilde{\pi}$ from $s$ to $t$ in an input graph $\widetilde{G}$ can be found by computing a shortest group Steiner paths $\pi$ from $s$ to $t$ in $\widetilde{G}$ when each vertex is in its own group $R_i$.
In both cases the cost of $\widetilde{\pi}$ is equal to the length of $\pi$.

\subsection{An algorithm with a running time of \texorpdfstring{$2^k k n^{O(1)}$}{2\^k k n\^O(1)}}
\label{sec:computing_group_Steiner_paths}

In this section we show how the group Steiner distances from $s$ to all $v \in V$ can be computed in time $O(2^k k |R|^2 + |R| (m + n \log n)) = 2^k k n^{O(1)}$. After running our algorithm, each shortest group Steiner path from $s$ to any $t \in V$ can be returned in time proportional to the path's length.

We start by computing all the shortest paths (w.r.t.\ the shortest path metric) from the vertices $u \in R \cup \{s\}$ towards all $v \in V$, along with the distances $d_G(u,v)$.
This can be done by a repeated invocation of Dijkstra's algorithm and requires an overall time of $O( |R| (m + n \log n) )$.

We now compute the length $\BD(s,r \mid \mathcal{C})$ of the shortest group Steiner path from $s$ to $r$ w.r.t.\ the groups in $\mathcal{C}$, for all $r \in R$ and all $\mathcal{C} \subseteq \mathcal{R}$.

According to our definition, we have $\BD(s, r \mid \emptyset) = d_G(s,r)$.
Using \Cref{lemma:group_steiner_path_decomposition}, we obtain the following recursive formula  for $\mathcal{C} \neq \emptyset$:
\begin{equation}
    \label{eq:bd_from_s_to_r}
    \BD(s,r \mid \mathcal{C}) = \min_{r' \in \cup_{R_i \in \mathcal{C}} R_i} \BD(s,r' \mid \mathcal{C} \setminus \{ R_j \mid r' \in R_j\}) + d_G(r', r).
\end{equation}

Since the above formula can be evaluated in time $O(k |R|)$ (once $\BD(s,r \mid \mathcal{C}')$ is known for all $r \in R$ and $\mathcal{C}' \subset \mathcal{C}$), this immediately leads to a dynamic programming algorithm that computes $\BD(s,r \mid \mathcal{C})$ for all choices of $r \in R$ and $\mathcal{C} \subseteq \mathcal{R}$ in time $O( |R| (m + n \log n)+  2^k \cdot k  \cdot |R|^2 )$.

The group Steiner distance $\BD(s,v \mid \mathcal{R})$ from $s$ to a vertex $v \in V$ can then be found in time $O(|R|)$ by guessing the last required vertex $r$ in a shortest group Steiner path from $s$ to $v$, i.e.:
\begin{equation}
    \label{eq:bd_from_s_to_v}
    \BD(s,v \mid \mathcal{R}) = \min_{r \in R} \BD(s,r \mid \mathcal{R}) + d_G(r, v). 
\end{equation}

Therefore we can compute all group Steiner distances $\BD(s,v \mid \mathcal{R})$ for $v \in V$ in time $O(|R|n)$ with is dominated by the $O( |R| (m + n \log n)+  2^k \cdot k  \cdot |R|^2 )$ time complexity discussed above.

To return a shortest group Steiner path $\pi_v$ from $s$ to $v$ it suffices to additionally keep track of (i) the optimal choice of $r$ in \Cref{eq:bd_from_s_to_v} for each $v$, and (ii) the optimal choices of $r'$ in \Cref{eq:bd_from_s_to_r} for all subproblems encountered during the dynamic programming algorithm.
This allows to find at most $k$ required nodes in $\pi_v$ (in order) such that each $R_i$ contains at least one such node. By \Cref{lemma:group_steiner_path_decomposition} the portions of $\pi_v$ between any two consecutive required nodes are shortest paths in $G$, and all such shortest paths have been computed at the beginning of the algorithm.

\end{document}